\documentclass[twoside,11pt]{article}

\usepackage{epsfig}
\usepackage{amssymb}
\usepackage{natbib}
\usepackage{hyperref}
\usepackage{graphicx}
\usepackage{thumbpdf}
\usepackage{amsfonts,amstext,amsmath,amsthm}
\usepackage{mathtools}
\usepackage{dsfont}
\usepackage{accents}
\usepackage{color}
\usepackage{rotating}
\usepackage[utf8]{inputenc}
\usepackage{booktabs}
\usepackage{tikz}
\usepackage{multirow}
\usepackage{float}
\usepackage[labelfont=bf, font=small]{caption}
\usepackage[left=2.5cm, right=2.5cm, top=2.5cm, bottom=2.5cm]{geometry}
\usepackage{lscape}
\usepackage{tcolorbox}
\usepackage{enumitem}
\usepackage{subcaption}
\usepackage{crossreftools}

\usepackage{algorithm}
\usepackage{algpseudocode}

\setcitestyle{authoryear,open={(},close={)}}
\newcommand*{\doi}[1]{\href{http://dx.doi.org/#1}{doi:#1}}

\graphicspath{{Figures/}}

\newcommand{\aggr}{\beta}

\newcommand{\ivqr}{\calM_{\rm IVQR}}
\newcommand{\ivscm}{\calM_{\rm IVSCM}}

\DeclareMathOperator{\logistic}{Logistic}

\renewcommand{\hat}{\widehat}

\newcommand{\dd}{\,\mathrm{d}}
\newcommand{\1}{\mathds{1}}

\theoremstyle{definition}
\newtheorem{proposition}{Proposition}
\newtheorem{assumption}{Assumption}
\newtheorem*{proposition*}{Proposition}
\newtheorem{definition}{Definition}

\newtheorem{example}{Example}
\newtheorem{theorem}{Theorem}
\newtheorem{lemma}{Lemma}

\newtheorem{setting}{Setting}
\newtheorem{scenario}{Scenario}

\DeclareMathOperator{\dte}{DCE}
\DeclareMathOperator{\qte}{QCE}

\DeclareMathOperator{\dok}{DOK}

\usepackage{accents}

\newcommand{\norm}[1]{\left\lVert#1\right\rVert}
\newcommand{\onorm}[1]{\left\lvert#1\right\rvert}

\newcommand{\rY}{Y}

\newcommand{\ry}{y}

\newcommand{\basisy}{a}

\newcommand{\parm}{\vartheta}
\newcommand{\eparm}{\vartheta}

\newcommand{\ie}{{i.e.,}~}

\newcommand{\Prob}{\mathbb{P}}
\newcommand{\Ex}{\mathbb{E}}
\newcommand{\RR}{\mathbb{R}}

\usepackage{dsfont}

\newcommand{\given}{\mid}

 \DeclareMathOperator{\expit}{expit}

 \DeclareMathOperator*{\argmin}{{arg\,min}}

 \DeclareMathOperator{\ND}{Normal}
 
 \DeclareMathOperator{\UD}{Uniform}

 \DeclareMathOperator{\BD}{Bernoulli}

 \def \calD {\{0, 1\}}
 
 \def \calF {\mathcal F}
 \def \calG {\mathcal G}
 \def \calH {\mathcal H}

 \def \calM {\mathcal M}

 \def \calP {\mathcal P}
 \def \calQ {\mathcal Q}

 \def \calX {\mathcal X}
 \def \calY {\mathcal Y}
 \def \calZ {\mathcal Z}

\newcommand{\pkg}[1]{\texttt{#1}}

\newcommand{\proglang}[1]{\textsf{#1}}
\newcommand{\code}[1]{\texttt{#1}}

\DeclareMathOperator{\pdo}{do}

\usepackage{tikz}

\usetikzlibrary{positioning, calc, shapes.geometric, shapes.multipart,
  shapes, arrows.meta, arrows,
  decorations.markings, external, trees}
\tikzstyle{line} = [draw, -latex']
\tikzstyle{Arrow} = [
        thick,
        decoration={
                markings,
                mark=at position 1 with {
                        \arrow[thick]{latex}
} },
        shorten >= 3pt, preaction = {decorate}
        ]

\newcommand{\indep}{\perp\nolinebreak\hspace{-6pt}\perp} 
\mathtoolsset{showonlyrefs=true}

\setcitestyle{authoryear,open={(},close={)}}

\newcommand{\bcd}{\boldsymbol\cdot}

\title{\bf Instrumental Variable Estimation of
Distributional Causal Effects}

\author{Lucas Kook\textsuperscript{1}\thanks{%
E-mail:~\code{lucasheinrich.kook@gmail.com}.
LK carried out part of this work at
the Department of Mathematical Sciences,
University of Copenhagen, Denmark.
} ~and Niklas Pfister\textsuperscript{2}}
\date{%
\footnotesize
\textsuperscript{1}Institute for Statistics and Mathematics, Vienna University
of Economics and Business, Austria
\\\footnotesize
\textsuperscript{2}Department of Mathematical Sciences,
University of Copenhagen, Denmark
}

\begin{document}

\maketitle

\begin{abstract}%
Estimating the causal effect of a treatment on the entire response distribution
is an important yet challenging task. For instance, one might be interested in
how a pension plan affects not only the average savings among all individuals
but also how it affects the entire savings distribution. While sufficiently
large randomized studies can be used to estimate such distributional causal
effects, they are often either not feasible in practice or involve
non-compliance. A well-established class of methods for estimating average
causal effects from either observational studies with unmeasured confounding or
randomized studies with non-compliance are instrumental variable (IV) methods.
In this work, we develop an IV-based approach for identifying and estimating
distributional causal effects. We introduce a distributional IV model with
corresponding assumptions, which leads to a novel identification result for the
interventional cumulative distribution function (CDF) under a binary treatment.
We then use this identification to construct a nonparametric estimator, called
DIVE, for estimating the interventional CDFs under both treatments. We
empirically assess the performance of DIVE in a simulation experiment and
illustrate the usefulness of distributional causal effects on two real-data
applications.
\end{abstract}

\section{Introduction}

Unlike average causal effects, which only capture how a given treatment affects
the response on average, distributional causal effects capture the effect on the
entire response distribution, including the mean, variance and higher moments of
the response \citep{kneib2023rage}. Estimands beyond the conditional mean are of
interest in many applications, such as electricity price forecasting
\citep{marcjasz2023energy}, poverty research \citep{hohberg2021poverty},
epidemiology \citep{lohse2017continuous}, and the analysis of survival times
\citep{cox1972regression}. Estimating distributional causal effects is feasible
when data from randomized experiments are available, for instance, using
quantile \citep{koenker2017handbook} or distributional regression
\citep{kneib2013beyond}. However, in the presence of unobserved confounding, the
task of estimating distributional causal effects is considerably more
challenging and the available methods are limited.

Observational data, despite the presence of hidden confounding, may, under
appropriate assumptions, be used to draw causal conclusions
\citep{pearl2009causality,wooldridge2015introductory}. Often, the estimand is
the average causal effect, \ie the difference in the average response between
treated and untreated experimental units. One way to estimate average causal
effects from observational data is via instrumental variables (IVs), which
induce heterogeneity in the treatment assignment that is independent of both
the hidden confounders and only affects the response via the treatment
\citep{haavelmo1943statistical,angrist1996identification,imbens1997estimating}.
IV approaches for estimating average causal effects have been proposed for
non-linear additive noise models \citep{newey1990nonliniv}, censored and
truncated outcomes \citep{newey2001censored}, the nonparametric additive noise
model \citep{newey2003ivnonparm,newey2013nonparm}, semi-parametric cases
\citep{hansen2010ivflex}, and high-dimensional data and sparse causal effects
\citep{pfister2022identifiability}. For the case of non-linear additive noise
models, several kernel- \citep{singh2019kernel,muandet2020dual,zhang2020maximum}
and neural network-based \citep{bennett2019deep} estimators have been developed.
Most existing IV-based methods make use of conditional moment restrictions;
however, recent work \citep[e.g.,][]{Poirier2017, Dunker2021,
saengkyongam2022exploiting},
has proposed to use a stronger
identifiablity condition based on independence between the instruments and the
residuals. Our work expands on this line of work and also proposes an
independence-based identification strategy, which is sufficiently strong to
ensure identifiability even of distributional causal effects. The work closest
to ours is the instrumental variable quantile regression (IVQR) framework due to
\citet{chern2005ivqte}, which, similarly to us, considers a binary treatment and
absolutely continuous response. For continuous treatments,
\citet{imbens2009identification} propose a control function approach and, for
arbitrary treatments but binary instruments, \citet{chernozhukov2024estimating}
propose an approach based on distributional regression via copulas.

In this work, we present a novel identification result for distributional causal
effects in settings with an absolutely continuous response $Y$ under a binary
treatment $D$. We introduce a IV model based on structural causal models
\citep[SCMs,][]{pearl2009causality}, which we term IV-SCM and provide
assumptions under which the distributional causal effect is identifiable. The
resulting identification result (Theorem~\ref{thm:id}) is based on the fact that
the interventional probability integral transform (iPIT) residuals $F^*_D(Y)$
(i.e., the interventional CDFs under the observed treatment and evaluated at the
observed response) are uniformly distributed and independent of the instrument.
Based on this result, we propose the distributional instrumental variable
estimator (DIVE) which trades off uniformity of the iPIT residuals \citep[via
the Cram\'er--von Mises criterion][]{cramer1928} with its independence from the
instrument \citep[via the Hilbert-Schmidt independence
criterion,][]{gretton2007hsic}. DIVE is proven to be consistent, which is
empirically validated in a simulation study. In two 
real-data
applications, DIVE
leads to similar conclusions as IVQR (based on linear quantile regression) but
is more stable and avoids quantile crossing. Finally, we also theoretically
compare the proposed IV-SCM with the IVQR model. Importantly, if the quantile
causal effect is identifiable in an IVQR, then there exists a corresponding
IV-SCM for which our identification result holds
(Proposition~\ref{prop:equiv_final}).

The remainder of the paper is structured as follows. Section~\ref{sec:dce}
introduces IV-SCMs and defines distributional causal effects.
Section~\ref{sec:id} then turns to the assumptions required for identifying
distributional causal effects and contains our identification result.
Section~\ref{sec:ivqr} gives a theoretical comparison between IV-SCM and the
IVQR model. Section~\ref{sec:DIVE} formally introduces DIVE and provides the
explicit algorithm for hot to implement it Section~\ref{sec:emp} contains the
simulation study and two real-data applications. Finally, in
Section~\ref{sec:discussion}, we provide a discussion and outlook.

\section{Distributional causal effects}\label{sec:dce}

We consider distributional causal effects as contrasts between the
interventional CDF of treated and untreated populations (assuming a binary
treatment). Throughout this paper, we consider an absolutely continuous response
$Y \in \calY \subseteq \RR$ together with a binary treatment indicator $D \in
\calD$ in the presence of hidden confounders $H \in \calH$. Further, we assume
access to instrumental variables $Z \in \calZ$. We use SCMs, to define the
interventional CDF, but our results can also be phrased using other causal
models, such as potential outcome models (see Section~\ref{sec:ivqr}).

\begin{setting}[IV-SCM]\label{setting}
Let $\calZ\subseteq \RR^{d_Z}$, $\calH \subseteq \RR^{d_H}$, $\calY\subseteq\RR$
be measurable sets and let $(Z, H, D, Y) \in \calZ \times \calH \times \calD
\times\calY$ be random variables satisfying the following SCM:
\begin{center}
\begin{minipage}[b]{0.49\textwidth}
\begin{align*}
\begin{cases}
    Z \coloneqq N_Z\\
    H \coloneqq N_H\\
    D \coloneqq f(Z, H, N_D)\\
    Y \coloneqq g_D(H, N_Y)
\end{cases}
\end{align*}
\end{minipage}
\begin{minipage}[b]{0.49\textwidth}
\begin{tikzpicture}[node distance=1.0cm, <-> /.tip = Latex, -> /.tip = Latex,
    thick, roundnode/.style={circle, draw, inner sep=1pt,minimum size=7mm}, 
    squarenode/.style={rectangle, draw, inner sep=1pt, minimum size=7mm}]
\node [roundnode] (Z) {$Z$};
\node [right=of Z, roundnode] (D) {$D$};
\node [above right=of D, roundnode, fill=gray!20] (H) {$H$};
\node [below right=of H, roundnode] (Y) {$Y$};
\draw[->] (Z) -- (D);
\draw[->] (D) -- (Y);
\draw[->] (H) -- (Y);
\draw[->] (H) -- (D);
\end{tikzpicture}
\end{minipage}\\[10pt]
\end{center}
where $N_Z, N_H, N_D, N_Y$ are noise variables satisfying $(N_Z,N_D)\indep
N_Y\given N_H$ and $N_Z \indep N_H$,\footnote{For simplicity one can also assume
that $N_Z, N_H, N_D, N_Y$ are jointly independent. This is, however, slightly
more restrictive than necessary here.} and, for all $d \in \calD$, $g_d(N_H,
N_Y)$ is absolutely continuous. We denote by
$\mathcal{D}_n=\{(Z_1,D_1,Y_1),\ldots,(Z_n,D_n,Y_n)\}$ a sample of $n$
independent copies of $(Z, D, Y)$ from the above SCM.
\end{setting}
We only use the IV-SCMs to model the observed distribution and the two
do-interventions $\pdo(D=0)$ and $\pdo(D=1)$. Importantly, we do not use any of
the other induced interventional and counterfactual distributions and in fact do
not even require them to be correctly specified.

\begin{definition}[Interventional quantities]
Under Setting~\ref{setting}, for all $d\in\{0,1\}$, we define the
\emph{interventional cumulative distribution function (interventional CDF) under
treatment $d$} $F^*_d : \RR \to [0, 1]$, for all $y\in\RR$ by
\begin{equation*}
   F^*_d(y)\coloneqq\Ex[\1(\rY \leq y) \given \pdo(D = d)].
\end{equation*}
We refer to the random variable $F_D^*(Y)$ as the (population)
\emph{interventional probability integral transform (iPIT) residual}
\citep[after the observational PIT, see][]{rosenblatt1952pit}. We further define
for all $d\in\calD$ the \emph{interventional quantile function under treatment
$d$,} $Q_d^*:(0,1)\to\RR$ for all $\tau\in(0,1)$ by $Q_d^*(\tau) \coloneqq
\inf\{\ry\in\calY \mid F_d^*(y) \geq \tau\}$.
\end{definition}

\subsection{Distributional causal effects on different scales}

Average causal effects can only describe effects in the population average and
are unable to capture whether specific parts of the distribution of $Y$ are
affected differently. To capture such effects, one can use distributional causal
effects, which can be expressed on various different scales, with the difference
in CDFs being an intuitive example that can be interpreted as a risk difference
varying with the response (see Figure~\ref{fig:dte}\textsf{A} and the first
panel of~\textsf{B}). However, other intepretational scales, often implied by
parametric modelling assumptions (such as differences on the log-odds scale in
logistic regression), may be more familiar to practitioners. We define several
transformed distributional causal effect that capture, on different scales, how
the response distributions differ across the interventions $\pdo(D=0)$ and
$\pdo(D=1)$.

\begin{definition}[Transformed distributional causal effects]\label{def:tce}
Assume Setting~\ref{setting}. We define the following transformed distributional
causal effects:
\begin{itemize}
    \item The \emph{distributional causal effect}
        $\dte:\mathcal{Y}\rightarrow\mathbb{R}$ for all $y\in\mathcal{Y}$ by
        \begin{align}
            \dte(y)\coloneqq F^*_1(y) - F^*_0(y).
        \end{align}
    \item The \emph{quantile causal effect} $\qte:(0,1)\rightarrow\calY$ for
        all $\tau\in(0,1)$ by
        \begin{align}
            \qte(\tau)\coloneqq Q_1^*(\tau) - Q_0^*(\tau).
        \end{align}
    \item The \emph{Doksum causal effect}
        $\dok:\mathcal{Y}\rightarrow\mathbb{R}$ for all $y\in\mathcal{Y}$ by
        \begin{align}
            \dok(y)\coloneqq Q_0^*(F_1^*(y)) - y.
        \end{align}
    \item The \emph{logit distributional causal effect}
        $\operatorname{LogitCE}:\mathcal{Y}\rightarrow\mathbb{R}$ for all
        $y\in\mathcal{Y}$ by
        \begin{align}
            \operatorname{LogitCE}(y)\coloneqq
            \log\left(\frac{F_1^*(y)}{1-F_1^*(y)}\right) -
            \log\left(\frac{F_0^*(y)}{1-F_0^*(y)}\right).
        \end{align}
\end{itemize}
\end{definition}
The Doksum causal effect \citet{doksum1974} corresponds to taking vertical
differences between the interventional CDF for treated and untreated
observations as a function of the response. Assuming rank invariance (see
Section~\ref{sec:asmp}), DOK can be interpreted as the difference in potential
outcomes under treatment and control. Depending on the application at hand other
versions of transformed distributional effects could also be considered.

Our proposed estimator introduced in Section~\ref{sec:DIVE} is constructed to
estimate the interventional CDF under treatment and control and can thus be used
to construct plug-in estimators for any of the (transformed) distributional
causal effects. Figure~\ref{fig:dte} illustrates the different interpretational
scales for distributional causal effects on simulated data. In this example, the
typically made structural assumption of proportional odds \citep[see, e.g.,\
][]{lohse2017continuous}, corresponding to a constant logit distributional
causal effect, is not tenable and a nonparametric approach is needed to capture
the correct distributional causal effect.

\begin{figure}[t!]
\centering
\includegraphics[width=0.99\textwidth]{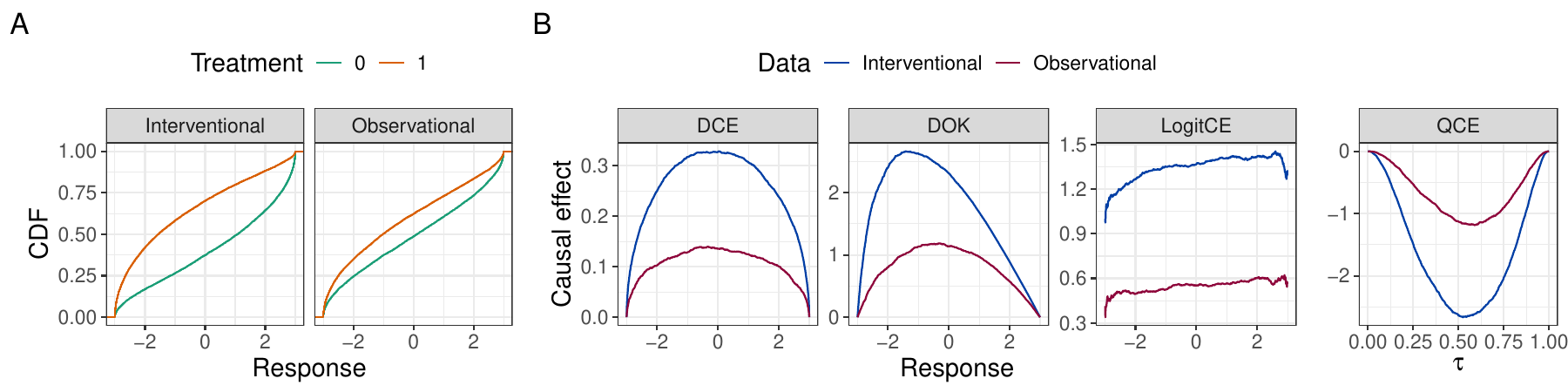}
\caption{%
Illustration of transformations of the distributional causal effect with
simulated (observational and interventional) data. \textsf{A}: Interventional
and observational CDFs of the response for the two treatment groups. \textsf{B}:
Distributional causal effects on different scales and computed both based on the
estimated interventional and observational CDFs, leading to correct and biased
effects, respectively. $\dte$ captures the difference in CDFs (vertical), while
$\dok$ is a difference on the original scale of the response (horizontal).
$\operatorname{LogitCE}$ shows that the treatment effect is non-constant, which
violates the proportional odds assumption in a simple logistic regression model.
}
\label{fig:dte}
\end{figure}

\section{Identifying distributional causal effects}\label{sec:dive}

In order to identify distributional causal effects, we require additional
assumptions on the IV-SCM. We introduce several assumptions in
Section~\ref{sec:asmp} and show that they are sufficient for identifiablity in
Theorem~\ref{thm:id}. Finally, in Section~\ref{sec:ivqr} we compare these
assumptions with the IVQR model and its assumptions and show that the
assumptions provided here are slightly weaker.

\subsection{Identification assumptions}\label{sec:asmp}

We begin with an assumption that constrains how the confounding can influence
the treatment assignment.

\begin{assumption}[Propensity-conditioned exclusion restriction]\label{asmp:pcu}
In Setting~\ref{setting}, let $\pi(H) \coloneqq \Ex[D \given H]$ denote the
treatment propensity and assume
\begin{align}
    Z \indep Y \given D, \pi(H).
\end{align}
\end{assumption}
This assumption is a slightly stronger version of the exclusion restriction
assumption, $Z \indep Y \given D, H$, used in classical IV
\citep{angrist1996identification}. Intuitively, Assumption~\ref{asmp:pcu} posits
that it is sufficient to control for the treatment propensity instead of the
hidden confounders themselves.
Assuming independent noise, Assumption~\ref{asmp:pcu} can also be interpreted in
terms adding of a node $\pi(H)$ that $d$-separates $H$ from $D$ (or $H$ from
$Y$) in the graph from Setting~\ref{setting}. 
Assumption~\ref{asmp:pcu} can be defended by arguing that the exclusion
restriction is fulfilled and the confounding mechanism depends only on the
treatment propensity. 

Next, we consider assumptions on the ranks of individual observations among the
treated and untreated group.

\begin{assumption}[Rank assumptions]\label{asmp:crs}
In Setting~\ref{setting}, let $\pi(H) \coloneqq \Ex[D \given H]$ denote the
treatment propensity conditional on $H$. We then assume one of the following
conditions.
\begin{enumerate}
    \item[] {(\crtcrossreflabel{RI}[asmp:crs:rinv])} \emph{Rank
    invariance}:
        \begin{align}\label{eq:rinv}
            F^*_D(Y)\indep D \given H, N_Y.
        \end{align}
    \item[] {(\crtcrossreflabel{RS}[asmp:crs:rsim])} \emph{Rank
    similarity}:
        \begin{align}\label{eq:rsim}
            F^*_D(Y)\indep D \given H.
        \end{align}
    \item[] {(\crtcrossreflabel{CRS}[asmp:crs:crs])} \emph{Conditional rank
    similarity}:
        \begin{align}\label{eq:crs}
            F^*_D(Y)\indep D \given \pi(H).
        \end{align}
\end{enumerate}
\end{assumption}
Rank invariance is equivalent to $F^*_1(g_1(H, N_Y)) = F^*_0(g_0(H, N_Y))$ being
true almost surely. Rank similarity and conditional rank similarity assume that,
after taking into account the hidden confounding or treatment propensity,
respectively, the distribution of the rank of an individual 
is the same regardless of whether this individual is treated or untreated. The
following result shows that conditional rank similarity is strictly weaker than
rank similarity and rank invariance. A proof is given in
Appendix~\ref{proof:prop:rankasmp}.

\begin{proposition}[Rank assumptions]\label{prop:rankasmp}
Under Setting~\ref{setting}, Assumption~\ref{asmp:crs}
\ref{asmp:crs:rinv} implies \ref{asmp:crs:rsim} and
\ref{asmp:crs:rsim} implies \ref{asmp:crs:crs}. None of the reverse implications
hold. 
\end{proposition}

Example~\ref{ex:holds} below, provides an example IV-SCM in which rank
invariance is violated, but rank similarity is not. The remaining examples prove
that none of the reverse implications hold are provided in
Examples~\ref{ex:nobreak}~and~\ref{ex:break} in Appendix~\ref{app:cexs}. They
may help the reader to gain additional intuition on the rank assumptions.

\begin{example}[Only rank invariance violated]\label{ex:holds}
Consider the IV-SCM given by
\begin{align*}
    H &\coloneqq N_H\\
    D &\coloneqq \1(N_D \geq 0.2 + 0.6 \1(H\geq 0))\\
    Y &\coloneqq H - D + D N_Y - (1 - D) N_Y,
\end{align*}
where $N_H, N_Y \sim \ND(0, 1)$, $N_D \sim \UD(0,1)$ jointly independent. Here,
it holds that $F_0 : \ry \mapsto \Phi(\ry/\sqrt{2})$ and $F_1 : \ry \mapsto
\Phi((1 + \ry)/\sqrt{2})$. We now discuss the three types of rank assumptions
mentioned above and summarize the example in Figure~\ref{fig:ex:holds}.
\begin{figure}
\centering
\includegraphics[width=\textwidth]{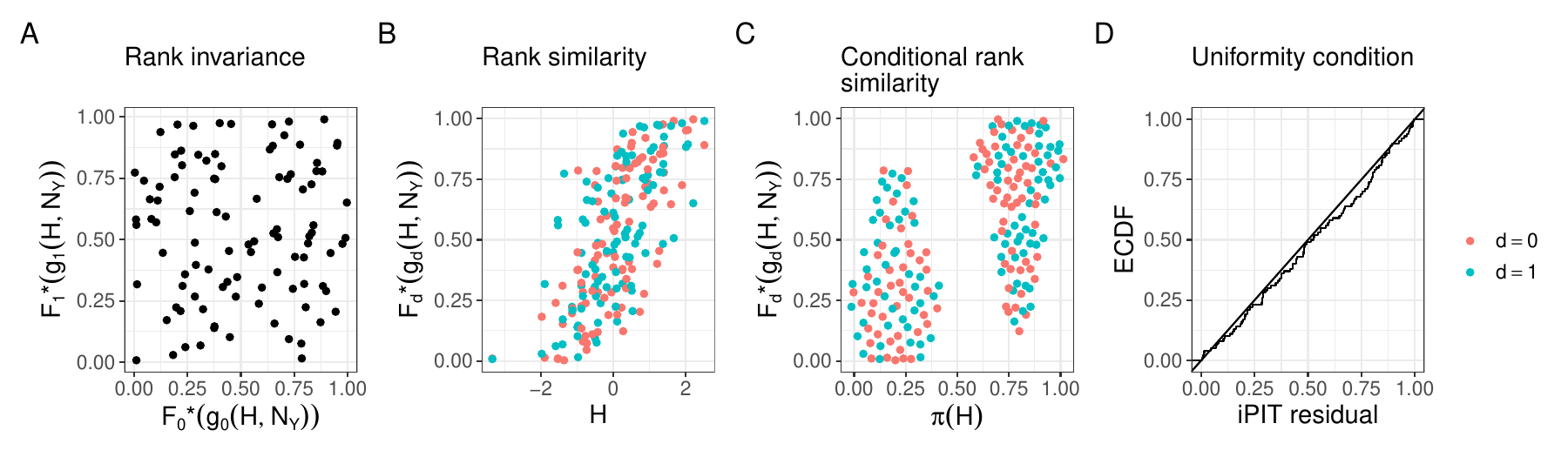}
\caption{%
Overview of the rank assumptions discussed in Example~\ref{ex:holds}.
A: Rank invariance does not hold.
B: Rank similarity holds.
C: Conditional rank similarity holds.
D: The iPIT residuals are uniform.
}\label{fig:ex:holds}
\end{figure}
Since $F_0(g_0(H, N_Y)) = \Phi((H + N_Y)/\sqrt{2})$ and $F_1(g_1(H, N_Y)) =
\Phi((H - N_Y) / \sqrt{2})$, rank invariance is violated.
Figure~\ref{fig:ex:holds}A shows the (non-degenerate) joint distribution of the
ranks $F_0^*(g_0(H,N_Y))$ and $F_1^*(g_1(H,N_Y))$. Rank similarity is fulfilled,
since for all $h \in \RR$, $\Phi((h + N_Y)/\sqrt{2}) \stackrel{d}= \Phi((h -
N_Y)/\sqrt{2})$. This is illustrated in Figure~\ref{fig:ex:holds}B, which plots
the conditional distributions of the iPIT residuals given the hidden confounder
$H$. Conditional rank invariance holds by Proposition~\ref{prop:rankasmp} and is
illustrated in Figure~\ref{fig:ex:holds}C. In Figure~\ref{fig:ex:holds}D, we see
that the iPIT residuals $D F_1^*(Y) + (1-D) F_0^*(Y) = D \Phi((H +
N_Y)/\sqrt{2}) + (1 - D)\Phi((H + 2N_Y)/\sqrt{3})$ are uniformly distributed.
\end{example}

Finally, we consider an assumption on the relationship between treatment and
instrument that slightly strengthens the ``relevance'' assumption
\citep{angrist1996identification}, which assumes that $D\not\indep Z$ and is
used for nonparametric identification of average causal effects.

\begin{assumption}[Relevance]\label{asmp:relevance}
In Setting~\ref{setting}, assume that for all measurable and bounded 
$m,k:\calH\rightarrow\mathbb{R}$ it holds that
\begin{equation*}
    m(H)+k(H)D\indep Z
    \quad\Longrightarrow\quad
    k(H) = 0 \quad \Prob_H\mbox{-almost surely.}
\end{equation*}
\end{assumption}

Assumption~\ref{asmp:relevance} requires that the instrument affects the joint
distribution of $D$ and $Y$ in a specific way. Under no hidden confounding, the
assumption is trivially fulfilled. If hidden confounding is present,
Assumption~\ref{asmp:relevance} is stronger than the conventional notion of
relevance (i.e.,~$D\not\indep Z$) and additionally enforces a type of
monotonicity (see Example~\ref{ex:relevance} and \eqref{eq:monotoneLR} below).
Although Assumption~\ref{asmp:relevance} is not testable directly, it implies
the classical notion of relevance, $D \not\indep Z$, which is
testable.\footnote{Assume relevance is violated, \ie $D \indep Z$ holds. Then,
for $m\equiv 0$ and $k\equiv 1$, $m(H)+k(H)D\indep Z$ but $k(H)=1$ almost
surely. Thus, Assumption~\ref{asmp:relevance} is violated.%
}%
We show in Example~\ref{ex:relevance} that this assumption is satisfied in
simple IV-SCMs.

\begin{example}\label{ex:relevance}
Consider an IV-SCM that fulfills $N_Z, N_H \sim \BD(0.5)$, $N_D \sim \UD(0, 1)$
which are jointly independent and $D \coloneqq f(Z, H, N_D) = \1(N_D \leq p(H,
Z))$. In this setup, Assumption~\ref{asmp:relevance} is satisfied if
$p(0,0)-p(0,1)\neq p(1,0)-p(1,1)$ and for all $z,h\in\{0,1\}$, $p(z,h)\neq 1$.
This type of condition on the probabilities is similar in spirit to the
monotonicity assumption in binary IV models, as it avoids cancellations in how
the effect of $Z$ propagates via $D$. This can be seen by considering the
conditional distribution of $H$ and $D$ given $Z$. A formal proof is given in
Appendix~\ref{proof:ex:relevance}.
\end{example}

\subsection{Identification result}\label{sec:id}

We are now ready to present the main identification result which also serves as
a basis for the estimator we propose in Section~\ref{sec:DIVE}. A proof is given
in Appendix~\ref{proof:thm:id}.

\begin{theorem}[Identification]\label{thm:id}
Assume Setting~\ref{setting}. Then, under Assumption~\ref{asmp:relevance} and
either Assumptions~\ref{asmp:pcu} and~\ref{asmp:crs}~\ref{asmp:crs:crs} or
Assumption~\ref{asmp:crs}~\ref{asmp:crs:rsim}, it holds for all $F_0, F_1 \in
\calF \coloneqq \{F : \RR \to [0, 1] \given F \mbox{ is a continuous CDF}\}$
that
\begin{align}\label{eq:id}
    F_D(Y) \indep Z \mbox{ and } F_D(Y) \sim \UD(0, 1)
    \iff
    \forall d\in\calD: \ F_d = F_d^*.
\end{align}
\end{theorem}

Thus, for all $d\in\calD$, the interventional CDF under treatment $d$ is
identified under either a weaker rank similarity assumption
(Assumption~\ref{asmp:crs} \ref{asmp:crs:crs}) but a stronger assumption on the
confounding mechanism (Assumption~\ref{asmp:pcu}), or a stronger rank similarity
assumption (Assumption~\ref{asmp:crs} \ref{asmp:crs:rsim}) but without
restricting the confounding mechanism. In both cases, relevance is required for
identification: If, for instance, $Z \indep (D, H, Y)$ holds, then
$F_D(Y) \indep Z$ holds for any $F$ which implies non-identifiability of the
interventional CDFs.
Assumptions~\ref{asmp:pcu} and \ref{asmp:crs} are not testable without further
restrictions on Setting~\ref{setting} and need to be argued based on domain
knowledge; while Assumption~\ref{asmp:relevance} has testable implications (see
the discussion below Assumption~\ref{asmp:relevance}).

\subsection{Relation to potential outcomes and the IVQR model}\label{sec:ivqr}

Estimating distributional causal effects using IVs has been considered before
using a formulation based on quantile functions and potential outcomes, called
IVQR
\citep{chernozhukov2004effects,chern2005ivqte,chern2006ivquantile,chern2008ivquantilerobust,chern2007ivquantileinference}.
However, the identification strategy employed in IVQR is fundamentally different
to the one we propose and use in Theorem~\ref{thm:id}. To illustrate this, we
first introduce the IVQR model and its identification strategy and then show
that if an IVQR model has an identifiable quantile causal effect, the model can
be expressed in terms of an IV-SCM (Setting~\ref{setting}) in which our
identifiability conditions hold. The IVQR model \citep{chern2005ivqte} is
formulated in terms of potential outcomes \citep{rubin1974estimating} and relies
on the conditions outlined in Setting~\ref{setting:ivqr}.\footnote{We have
adapted the notation/setting used in \citet{chern2005ivqte} to be more in line
with ours. In particular, we use the slightly stronger independence assumption
$(U(0),U(1),V) \indep Z$ instead of $(U(0),U(1)) \indep Z$ and express the rank
similarity condition in terms of conditional independence rather than
conditional distributions.}

\begin{setting}[IVQR model]\label{setting:ivqr}
  There exists a distribution over the variables
  \begin{equation}
    \label{eq:potential_outcomes}
    (Z,D,U(0), U(1), Y(0),
    Y(1))\in\mathcal{Z}\times\{0,1\}\times
    [0,1]\times[0,1]\times\mathcal{Y}\times\mathcal{Y},
  \end{equation}
  satisfying the following conditions:
  \begin{enumerate}
    \setlength{\itemsep}{0pt}
  \item (Potential outcomes). For all $d \in \calD$, $U(d) \sim \UD(0,1)$ and
    there exists a strictly increasing function $q_d : [0, 1] \to \RR$, such
    that $Y(d) = q_d(U(d))$.
    \label{ivqr:po}
  \item (Selection). There exists a measurable
    $\delta : \RR^{d_Z} \times \RR^{d_V} \to \calD$ and random
    vector $V \in \RR^{d_V}$, such that $D = \delta(Z, V)$.
    \label{ivqr:sel}
  \item (Independence). It holds that 
      $(U(0), U(1), V)\indep Z$. \label{ivqr:ind} 
  \item (Rank similarity). \label{ivqr:rank} It holds that
    $U(D)\indep D\given V, Z$.
    \item (Consistency, no interference, i.i.d.). The observed data
    $\mathcal{D}_n \coloneqq \{(Z_1, D_1, Y_1),\ldots,(Z_n,D_n,Y_n)\}$ are
    i.i.d.\ copies of $(Z, D, Y(D))$. 
    \label{ivqr:obs}
\end{enumerate}
\end{setting}

Condition~\ref{ivqr:po} ensure the existence of a structural quantile function
under treatment $d$, for all $d\in\calD$, that takes the endogenous noise $U(d)$
as input. Since for all $d \in \calD$, $U(d)$ is uniform, it can be thought of
as a rank of an individual under treatment $d$. Rank similarity corresponds to
the conditional distribution of $U(d)$ given $V$ and $Z$ to be the same for all
$d$. For our identification strategy, we employed conditional rank similarity
(Assumption~\ref{asmp:crs} \ref{asmp:crs:crs}), which was shown to be strictly
weaker than both rank invariance and rank similarity
(Proposition~\ref{prop:rankasmp}). Condition~\ref{ivqr:sel} assumes the
existence of a structural equation for the treatment assignment, where $\nu$
plays the role of $(H, N_D)$. Condition~\ref{ivqr:obs} postulates consistency,
no interference and that the data consist of i.i.d.\ copies of the observable
$(Z, D, Y(D))$.

The identification strategy used for IVQR models \citep{chern2005ivqte} is based
on the fact that under Setting~\ref{setting:ivqr}, it holds for all $\tau \in
(0, 1)$ that
\begin{equation*}
    \Prob(Y \leq q_D(\tau) \given Z) = \tau,
\end{equation*}
that is, the conditional CDF of the observed $Y$ given $Z$ evaluated at the
structural quantile functions under the observed treatment $D$, and any $\tau$
does not depend on $Z$ and is equal to $\tau$. This is fundamentally different
from our condition in which the population iPIT residual $F^*_D(Y)$ is shown to
be uniform and independent of $Z$.

Proposition~\ref{prop:equiv_estimand} below establishes that for each IVQR model
there exists an IV-SCM from Setting~\ref{setting} which fulfills rank similarity
as defined in Assumption~\ref{asmp:crs} \ref{asmp:crs:rsim} with equivalent
observational and interventional distributions.

\begin{proposition}\label{prop:equiv_estimand}\label{prop:equiv_final}
  Let $M$ be an IVQR model as in Setting~\ref{setting:ivqr} and denote
  by $P^M$ the distribution over the potential outcomes
  \eqref{eq:potential_outcomes}. Then, there exists an IV-SCM
  $\mathcal{G}(M)$ as in Setting~\ref{setting} satisfying
  Assumption~\ref{asmp:crs}~\ref{asmp:crs:rsim}, such that
  \begin{equation*}
    P^M_{Z, D, Y(D)}=P^{\mathcal{G}(M)}_{Z,
      D, Y},
    \quad
    P^M_{Y(0)}=P^{\mathcal{G}(M);\pdo(D=0)}_{Y}
    \quad\text{and}\quad
    P^M_{Y(1)}=P^{\mathcal{G}(M);\pdo(D=1)}_{Y}.
  \end{equation*}
\end{proposition}
A proof is given in Appendix~\ref{proof:prop:equiv_estimand}.
Proposition~\ref{prop:equiv_estimand} in particular implies, given appropriate
relevance conditions, that, for all $d\in\calD$, whenever the structural
quantile function under treatment $d$ in an IVQR model is identifiable, there
exists a corresponding IV-SCM in which the interventional CDF under treatment
$d$ is identified as well. The assumption of rank similarity can even be
weakened slightly to conditional rank similarity in the IV-SCM setting (see
Assumption~\ref{asmp:crs}). \citet{chern2005ivqte} state an explicit relevance
condition for point identification when $Z, D \in \calD$ based on the impact of
the instrument on the joint density of $D$ and $Y$ conditional on $Z$, denoted
by $f_{Y,D \given Z} : \RR \times \calD\times\calD \to [0,\infty)$. Namely, it
is assumed that for all $y \in \calY$,
\begin{align}\label{eq:monotoneLR}
\frac{f_{Y,D \given Z}(y, 1 \given 1)}{f_{Y,D\given Z}(y, 0 \given 1)}
>
\frac{f_{Y,D \given Z}(y, 1 \given 0)}{f_{Y,D\given Z}(y, 0 \given 0)},
\end{align}
which is referred to as a ``monotone likelihood ratio condition''. This notion
of relevance is distinct from the independence condition in
Assumption~\ref{asmp:relevance}.

\section{Distributional IV estimator}\label{sec:DIVE}

Using the identification strategy (Theorem~\ref{thm:id}), we propose an
estimation procedure based on uniformity of the empirical iPIT residuals and
their independence of the instrument. In Definitions~\ref{def:divl}
and~\ref{def:ipitiv} we present a loss function and an estimator which, for all
$d \in \calD$, consistently estimates the interventional CDF under treatment $d$
(Proposition~\ref{prop:consistency}).

\begin{definition}[Distributional IV loss]\label{def:divl}
Let $\lambda > 0$ be a constant, $\aggr : \RR_+ \times \RR_+ \to \RR_+$ a
Lipschitz-continuous aggregation function such that $\aggr(a, b) = 0$ if and
only if $(a, b)=(0,0)$ and $k_R:[0,1]\times[0,1]\rightarrow\mathbb{R}$ and
$k_Z:\mathcal{Z}\times\mathcal{Z}\rightarrow\mathbb{R}$ two positive definite
kernels. Moreover, let $\mathcal{D}_n=\{(\rY_i, D_i, Z_i)\}_{i=1}^n$ be data
from Setting~\ref{setting}. Then, for all pairs of CDFs $(F_0, F_1) \in \{F :
\RR \to [0, 1] \given F \mbox{ continuous CDF}\}$, we define the distributional
IV loss
\[
\ell_\lambda(F_0, F_1; \mathcal{D}_n) \coloneqq \aggr\bigg(
\widehat{\operatorname{CvM}}((F_{D_i}(Y_i))_{i=1}^n),
\lambda
\hat{\operatorname{HSIC}}((F_{D_i}(Y_i))_{i=1}^n, (Z_i)_{i=1}^n; k_R, k_Z)
\bigg),
\]
where $\widehat{\operatorname{CvM}}((F_{D_i}(Y_i))_{i=1}^n)$ denotes the
empirical Cram\'er--von Mises criterion w.r.t.\ the standard uniform distribution
\citep[see Definition~\ref{def:cvm:emp} in Appendix~\ref{app:aux},][]{cramer1928}
and $\hat{\operatorname{HSIC}}(\bcd, \bcd ; k_R, k_Z)$ denotes the empirical
Hilbert-Schmidt Independence criterion \citep[see Definition~\ref{def:hsic:emp}
in Appendix~\ref{app:aux},][]{gretton2007hsic} based on kernels $k_R$ and
$k_Z$.
\end{definition}

The CvM and HSIC terms entering the distributional IV loss enforce uniformity of
the empirical iPIT residuals and their independence of the instruments,
respectively: Firstly, the population CvM (Definition~\ref{def:cvm:pop} in
Appendix~\ref{app:aux}) is zero if and only if it's argument is uniformly
distributed on $(0,1)$ (Lemma~\ref{lem:cvm}).
The empirical Cram\'er--von Mises criterion can be used as a nonparametric test
of whether $(X_i)_{i=1}^n$ follow a standard uniform distribution
\citep{cramer1928}. 
Secondly, the population HSIC (see Definition~\ref{def:hsic:pop} in
Appendix~\ref{app:aux}) using characteristic kernels is zero if and only if the
arguments are independent. The choice of kernel can be made depending on the
sample space of $Z$. For continuous $Z$, we choose to use a Gaussian
kernel and for discrete $Z$ a discrete kernel. For the ranks, we choose
to use a Gaussian kernel after transforming the ranks with the standard normal
quantile function $\Phi^{-1}$. The empirical HSIC is formally defined in
Definition~\ref{def:hsic:emp} in Appendix~\ref{app:aux}.

In the following we choose to use the aggregation function $\aggr(a, b)
= a + b$. Other options are also possible, however the choice did not strongly
affect our empirical results so we do not consider this further (see
Appendix~\ref{app:sigma} for results using $\aggr(a, b) = \max\{a,b\}$). For the
case of normal interventional CDFs under both treatment $0$ and $1$ and with
equal variances, the loss landscape is shown for different values of $\lambda$
and interventional means $\mu_0, \mu_1$ in Figure~\ref{fig:loss}.

\begin{figure}[t!]
\centering
\includegraphics[width=\textwidth]{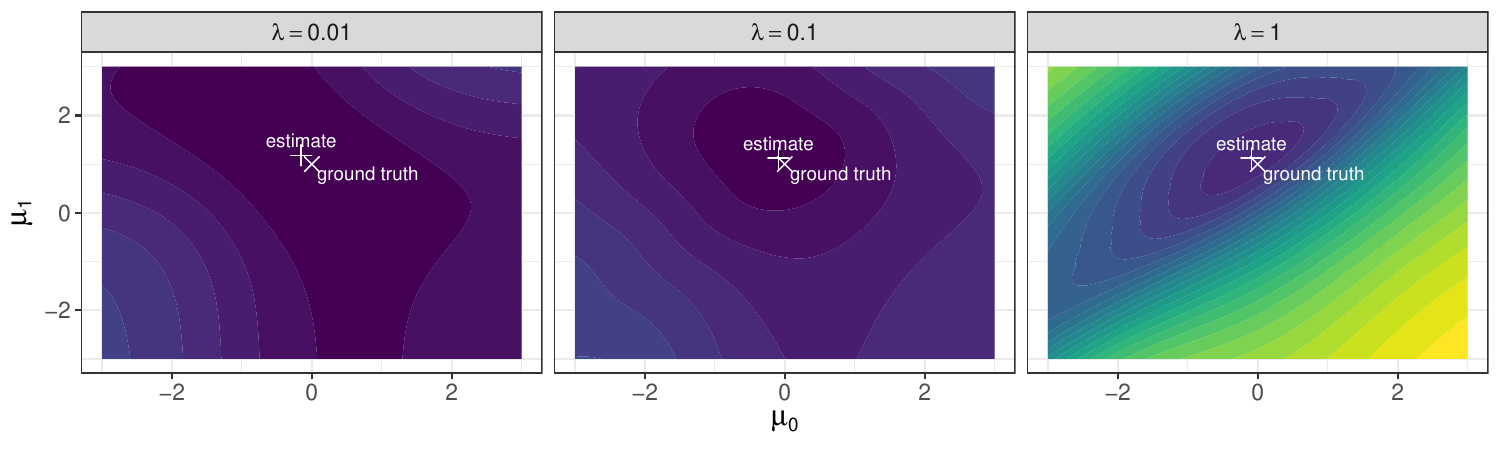}
\caption{%
Loss-landscape for the linear IV model in terms of the interventional means
$\mu_d$, $d \in \calD$ and for $\lambda \in \{0.01, 0.1, 1\}$. 
}\label{fig:loss}
\end{figure}

We parametrize the CDFs $F_0$ and $F_1$, using Bernstein polynomials of order
$M$. More formally, define $\basisy_M : [0, 1] \to \RR^{M+1}$ such that for all
$j\in\{1,\ldots,M+1\}$ and all $y\in[0,1]$ it holds that $(a_M)_j(y) =
\binom{M}{j} y^{j-1}(1-y)^{M-j+1}$ and define the parameter space $\Theta_M
\coloneqq \{\parm \in \RR^{M+1} \given \eparm_1 \leq \eparm_2 \leq \dots \leq
\eparm_{M+1}\}$. Then, due to the linear constraints in $\Theta_{M}$, the basis
expansion $\basisy_{M}(y)^\top\parm$ for any $\parm\in\Theta_{M}$ is monotone
increasing \citep{farouki2012bernstein}.
Since Bernstein polynomials can approximate any continuous function on the
interval $[0, 1]$, they can be used to parameterize the CDF of any bounded
absolutely continuous real-valued random random variable
\citep{farouki2012bernstein}.
We can therefore parameterize the interventional CDF under treatment $d$, for
all $d \in \calD$ and $\parm_d \in \Theta_{M}$, as $F_d : y \mapsto
\Phi(a_{M}(y)^\top\parm_d)$, where $\Phi$ denotes the standard normal CDF. Using
this parametrization together with the distributional IV loss we can define our
proposed estimator.

\begin{definition}[DIVE]\label{def:ipitiv}
Let $\mathcal{D}_n$ be data as in Setting~\ref{setting}, $M\in\mathbb{N}$ and $B>0$ fixed constants
and assume $\calY=[L, U]$ for some $L,U\in\mathbb{R}$. Then, for
$\mathcal{F}_B\coloneqq\{F:\mathbb{R}\rightarrow [0,1]\mid \exists
\vartheta\in\Theta_M \text{ with }\norm{\vartheta}_{\infty}<B \text{ s.t.\ }\forall y\in\mathbb{R}:\,
F(y)=\Phi(\basisy_{M}(\tfrac{y-L}{U-L})^\top\parm)\}$, we define the
distributional IV estimator (DIVE) as 
\begin{align}
    (\hat F_0^{n,\lambda}, \hat F_1^{n,\lambda}) \coloneqq
    \argmin_{(F_0,F_1)\in\calF_B^2} \ell_\lambda(F_0, F_1; \mathcal{D}_n).
\end{align}
\end{definition}

The choice of $\Phi$ in our parameterization influences the tail behavior of
DIVE. In our software, we implement several alternative choices to $\Phi$,
including the standard logistic, minimum and maximimum extreme value CDFs. In
practice, we suggest to cross-validate the choice. While we choose $\Phi$ and
Bernstein polynomials to parameterize CDFs for computational convenience, other
approaches based on nonparametric shape-constrained regression could be used as
well \citep{groeneboom2014nonparametric}.

As long as we assume that, for all $d\in\calD$, the true interventional CDF
under treatment $d$ can be parametrized using the same Bernstein polynomials,
DIVE converges with increasing sample size to the interventional CDF. A proof is
given in Appendix~\ref{proof:consistency}

\begin{proposition}[Consistency]\label{prop:consistency}
Assume Setting~\ref{setting}, Assumption~\ref{asmp:relevance} and either
Assumptions~\ref{asmp:pcu} and \ref{asmp:crs}~\ref{asmp:crs:crs} or
Assumption~\ref{asmp:crs} \ref{asmp:crs:rsim}. Furthermore assume there exists
$M\in\mathbb{N}$, $B>0$ and $\vartheta_0^*,\vartheta_1^*\in\Theta_M$ such that
for all $y\in\calY$ and $d\in\calD$ it holds that
$F_d^*(y)=\Phi(a_M(\tfrac{y-L}{U-L})^\top\vartheta_d^*)$ and $F_d^*\in\calF_B$.
Then, it holds for all $\lambda>0$ and $d\in\calD$ that
\begin{align}
    \sup_{y\in\calY} \,\lvert
    \hat{F}_d^{n,\lambda}(y) - F_d^*(y)
    \rvert \xrightarrow{\,p\,} 0 \quad\text{as $n\rightarrow\infty$},
\end{align}
where $\hat{F}_d^{n,\lambda}$ denotes the DIVE for $\calF_B$
(Definition~\ref{def:ipitiv}) using fixed bounded, continuously differentiable,
characteristic kernels $k_R$ and $k_Z$ with bounded derivatives.
\end{proposition}

The procedure we use to estimate the interventional CDF under treatment $d$ for
$d\in\calD$, is given in Algorithm~\ref{alg:dive}. The idea is to estimate the
CDFs as in Definition~\ref{def:ipitiv} for a fixed penalty parameter $\lambda$,
then test whether the resulting iPIT residuals satisfy uniformity and
independence from $Z$ and finally either terminate the algorithm if neither test
rejects (at the pre-specified level $\alpha$) or repeat the same steps but with
an updated value for $\lambda$. Basing the stopping criterion on tests of
uniformity and independence is a heuristic and does not provide frequentist
error control. More formally, we initialize $\lambda$ (\code{initialize\_lambda}
routine) such that the two terms entering the loss function in
Definition~\ref{def:divl} are of the same size for the conditional distribution
functions (i.e.,~ $\hat\Prob(Y \leq y \given D = d)$) estimated via maximum
likelihood and using the same Bernstein basis as for DIVE. We then estimate the
interventional CDFs under treatment $d$, for all $d\in\calD$, for $\lambda$ as
in Definition~\ref{def:ipitiv}, evaluate the iPIT residuals and compute
$p$-values for both a Cram\'er--von Mises (CvM) test \citep{cramer1928}, which
tests for uniformity of the iPIT residuals, and an HSIC independence test
\citep{gretton2007hsic}, which tests for independence between $Z$ and the iPIT
residuals. If both $p$-values exceed a pre-specified level $\alpha$, the
procedure stops and the interventional CDFs are returned. Otherwise, if the
$p$-value of the HSIC test is lower (higher) than that of the CvM test,
$\lambda$ is increased (decreased) with a heuristic and the same procedure is
run again. In the case that the procedure does not terminate after a fixed
number of $T$ iterations, the last estimated CDFs and a warning of
non-convergence at level $\alpha$ are returned. Failure to converge, could
either be due to a type-I error in the tests, exceeding the computational budget
(maximum number of updates), in which case the algorithm could be re-started
with a different value for $\lambda$ and a larger maximum number of updates, or
at least one of the identifying assumptions (Theorem~\ref{thm:id}) are violated
(see Section~\ref{sec:discussion}).

\begin{algorithm}
\caption{\small
Distributional Instrumental Variables Estimation.
}\label{alg:dive}
\begin{algorithmic}[1]
\Require Data $\mathcal{D}_n$, maximum number of restarts $T \in \mathbb{N}$,
significance level $\alpha \in (0, 1)$, sequence of step sizes
$\nu_1,\ldots,\nu_T>0$.
\State $t \gets 1$ 
\State $\lambda \gets$ \code{initialize\_lambda(}$\mathcal{D}_n$\code{)}
\While{$t \leq T$}
\State Obtain $(\hat{F}_0^{n,\lambda}, \hat{F}_1^{n,\lambda})$ as in 
    Definition~\ref{def:ipitiv} using SGD
\State $R_n \gets (\hat{F}^{n,\lambda}_{D_i}(Y_i))_{i=1}^n$
\State $p_U\gets$ $p$-value of CvM test for uniformity of $R_n$
\State $p_I\gets$ $p$-value of HSIC test for independence of $R_n$ and $(Z_i)_{i=1}^n$ with kernels $k_R$, $k_Z$
\If{$\min(p_U, p_I) > \alpha$} 
\State \Return{$(\hat{F}_0^{n,\lambda}, \hat{F}_1^{n,\lambda})$} 
\EndIf
\State $\lambda \gets \lambda (1 + \nu_t)^{2\1(p_I \leq p_U) - 1}$ 
\State $t \gets t + 1$
\EndWhile
\State 
\Return{$(\hat{F}_0^{n,\lambda}, \hat{F}_1^{n,\lambda})$ and 
`\textit{Warning:} DIVE did not converge.'} 
\end{algorithmic}
\end{algorithm}

If Algorithm~\ref{alg:dive} fails to terminate within the specified number of
restarts, it could indicate that some of the identifying assumptions in
Theorem~\ref{thm:id} are not satisfied. For instance, if $Z$ is not a valid
instrument ($N_Z \not\indep N_Y \given N_H$ or $N_Z \not\indep N_H$ in
Setting~\ref{setting}) or if (conditional) rank similarity
(Assumption~\ref{asmp:crs}) is violated, even the population iPIT residuals
could either be dependent on the instrument or not uniformly distributed,
leading the test to reject. Therefore, if the algorithm does not terminate
within the specified number of restarts, one should treat the resulting estimate
with care. 

\subsection{Computational details}

DIVE is implemented in the \proglang{R} language for statistical computing
\citep{pkg:base} and publicly available at
\url{https://github.com/LucasKook/dive}. The implementation of DIVE relies on
Bernstein and nonparametric transformation models trainable via gradient descent
implemented in \pkg{deeptrafo} \citep{pkg:deeptrafo}, using \pkg{tensorflow}
\citep{pkg:tensorflow} and the Adam optimizer \citep{Kingma2015adam}. The HSIC
criterion and test are implemented in \pkg{dHSIC} \citep{pkg:dHSIC} and the CvM
test is implemented in \pkg{goftest} \citep{pkg:goftest}. We use linear IVQR
implemented in the \pkg{IVQR} \proglang{R}~package available at
\url{https://github.com/yuchang0321/IVQR}. 

\section{Numerical experiments}\label{sec:emp}

\subsection{Simulation study}\label{sec:sim}

We investigate the following simulation scenarios within Setting~\ref{setting}
and describe comparator methods and evaluation metrics below. We consider sample
sizes $n \in \{100, 200, 400, 800, 1600\}$ with $M = 50$ and, to run
Algorithm~\ref{alg:dive}, we use $\alpha = 0.1$, $T = 10$, gradient descent
using Adam with learning rate $0.1$, callbacks to reduce the learning rate
(patience 20, tolerance $0.001$), stop early (patience 60, tolerance $0.0001$)
and step sizes $\nu_t = \nu / t$ for $\nu = 5$ and $t \in \{1, \dots, 10\}$. The
observational and interventional CDFs under treatment $d$, for all $d\in\calD$,
for each scenario are depicted in Figure~\ref{fig:simscen} in
Appendix~\ref{app:additional}. All scenarios fulfill Assumptions~\ref{asmp:crs}
\ref{asmp:crs:rsim}.

\begin{scenario}[Linear additive noise]\sloppy
Consider an IV-SCM with $N_Z, N_H, N_D \sim \logistic(0,1)$ mutually
independent, $f(H, Z, N_D) \coloneqq \1(4Z + 4H \geq N_D)$, $g_0(H, N_Y)
\coloneqq -10 + 6H$ and $g_1(H, N_Y) \coloneqq -2 + 6H$.
\end{scenario}

\begin{scenario}[Linear additive heteroscedastic noise]\sloppy
Consider an IV-SCM with $N_Z, N_H, N_D, N_Y \sim \logistic(0,1)$ mutually
independent, $f(H, Z, N_D) \coloneqq \1(4Z + 4H \geq N_D)$, $g_0(H, N_Y)
\coloneqq 6H + H N_Y$ and $g_1(H, N_Y) \coloneqq 16 + 6H + H N_Y$.
\end{scenario}

\begin{scenario}[Non-linear non-additive noise]\sloppy
Consider Setting~\ref{setting} with $N_Z, N_H, N_D \sim \logistic(0,1)$ mutually
independent, $f(H, Z, N_D) \coloneqq \1(4Z + 4 H \geq N_D)$, $g_0(H, N_Y)
\coloneqq \log(1 + \exp(18 + 6H))$ and $g_1(H, N_Y) \coloneqq \log(1+\exp( 26 +
6H))$.
\end{scenario}

\begin{scenario}[Linear additive noise with crossing]\sloppy
Consider an IV-SCM with $N_Z, N_H, N_D \sim \logistic(0,1)$ mutually
independent, $f(H, Z, N_D) \coloneqq \1(4Z + 4 H \geq N_D)$, $g_0(H, N_Y)
\coloneqq 6 H$ and $g_1(H, N_Y) \coloneqq 9 H - 6$.
\end{scenario}

We use a Bernstein polynomial of the same order as DIVE to estimate the
conditional CDF (CCDF) of $Y$ given $D$ to compare against the observational
conditional distribution. Due to the unobserved confounding this method is
biased.

Given a sample $\{(Z_i, D_i, Y_i)\}_{i=1}^n$ from one of Scenarios~1--4 above,
to quantify the discrepancy between the true interventional CDFs $(F^*_0,
F_1^*)$ and DIVE $(\hat{F}_0^{n,\lambda}, \hat F_1^{n,\lambda})$, we consider
the the mean squared error, which is given by 
\begin{align}
    \operatorname{MSE}(\hat{F}_0^{n,\lambda}, \hat{F}_1^{n,\lambda})\coloneqq
    \frac{1}{n} \sum_{i=1}^n (F_{D_i}^*(Y_i)-\hat{F}_{D_i}^{n,\lambda}(Y_i))^2,
\end{align}
and the maximum absolute error,
\begin{align}
    \operatorname{MAE}(\hat{F}^{n,\lambda}_0, \hat{F}^{n,\lambda}_1)
    \coloneqq\max_{i \in \{1, \dots, n\}} \lvert F_{D_i}^*(Y_i) -
    \hat{F}^{n,\lambda}_{D_i}(Y_i) \rvert,
\end{align}
which measures the largest distance between the true interventional CDFs and
the estimated CDFs.

Figure~\ref{fig:simres} summarizes the performance of DIVE and CCDF in terms of
MSE and MAE for increasing sample sizes and the scenarios described above. In
all scenarios, the estimation error for DIVE decreases linearly (on a
double-logarithmic scale). In contrast, for the CCDF the estimation error
converges to a constant, indicating that it converges to the true conditional
CDFs which are different from the interventional CDFs due to the unobserved
confounding.

\begin{figure}[t!]
\centering
\includegraphics[width=0.9\textwidth]{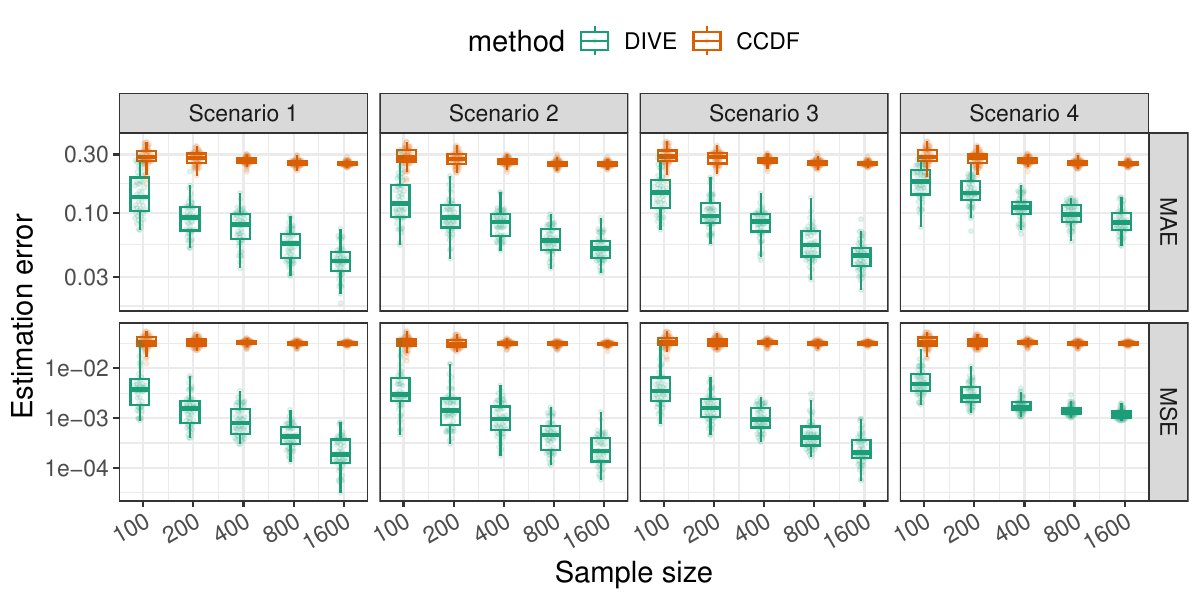}
\caption{%
Results for the simulation scenarios described in Section~\ref{sec:sim} for DIVE
and CCDF in terms of mean squared error and maximum absolute error compared to
the true interventional CDFs.
}\label{fig:simres}
\end{figure}

\subsection{Distributional causal effect of residing in a metropolitan area on
wage}\label{sec:schooling}

We apply DIVE to the \code{SchoolingReturns} dataset \citep{card1993using}
contained in the \proglang{R}~package \pkg{ivreg}. We model the distribution of
the log-transformed wage depending on whether an individual resided in a
standard metropolitan statistical area in 1976. The instrument is the binary
indicator of whether or not an individual grew up near a four-year college.
We use this classical IV dataset to showcase DIVE in a setting in which it can
be compared with competing methods. We do not aim to make any
subject matter claims about true causal effects. IV analysis requires careful
choice of instruments for which the identifying assumptions are thought to be
met and estimated causal effects ought to be interpreted as evidential.
We use the same settings for DIVE as in Section~\ref{sec:sim} but with $M=20$.
Figure~\ref{fig:schooling} summarizes the results in terms of the estimated CDFs
(\textsf{A}), uniformity of the empirical iPIT residuals and their independence
of the instrument (\textsf{B}), the estimated distributional causal effect
(\textsf{C}) and a comparison with two-stage least squares (2SLS) and a linear
model for estimating the average causal effect (\textsf{D}). 

\begin{figure}[t!]
\centering
\includegraphics[width=0.99\textwidth]{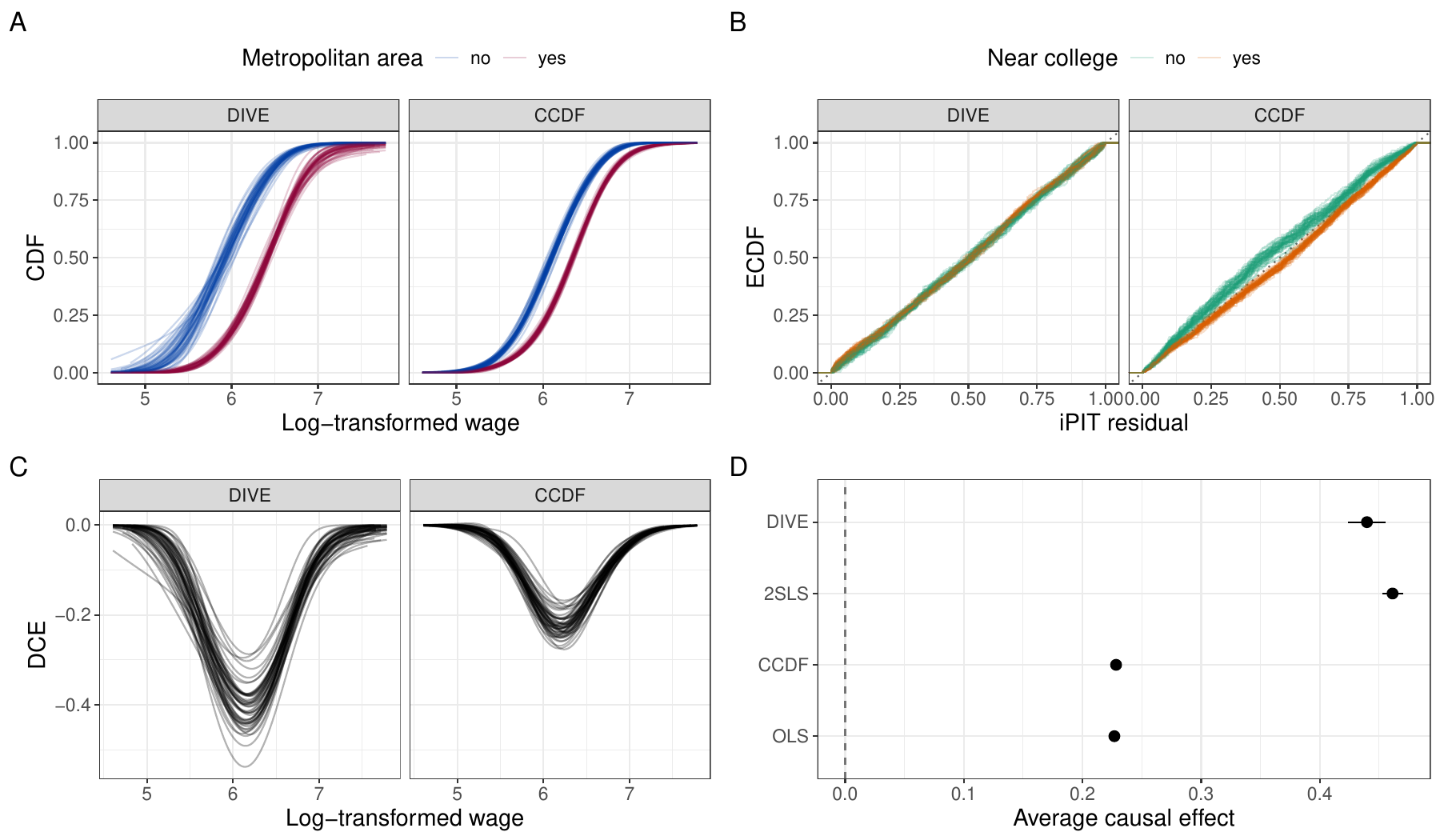}
\caption{%
Results for the schooling return dataset from Section~\ref{sec:schooling} for 50
subsamples of the data of size 1000. \textsf{A}: Estimated CDFs of
log-transformed wage by whether an individual resides in a
metropolitan area or not for DIVE and CCDF. \textsf{B}: ECDFs of the empirical
(interventional) PIT residuals stratified by the instrument (whether a person
grew up close to a four-year college or not). \textsf{C}: Distributional causal
effect estimated via DIVE and CCDF. \textsf{D}: Comparison of the estimated
average causal effect for DIVE, CCDF, OLS and 2SLS (obtained via numerical
integration for DIVE and CCDF). Standard errors are based on 50 subsamples.
}\label{fig:schooling}
\end{figure}

DIVE estimates a larger distributional causal effect for all wages compared to
the observational conditional distribution, which also translates into a larger
estimated average causal effect.
The iPIT residuals based on the fitted DIVE
satisfy that both a CvM test for uniformity and an HSIC test for independence to
the instrument (growing up near a four-year college) are not rejected.
Moreover, while DIVE and 2SLS agree in terms of direction and magnitude of the
average causal effect, the distributional causal effect estimates shows that the
effect is larger for lower wages. This illustrates how distributional causal
effects can provide further insights into the effect of living in a metropolitan
area on wages that cannot be captured by considering the average causal effect.

The estimated interventional CDFs can be inverted to obtain an estimate of the
quantile causal effect, which allows the comparison of DIVE and IVQR.
Figure~\ref{fig:qce} (a) shows the estimated QCEs for DIVE and IVQR. Overall,
there is large agreement between IVQR and DIVE over 50 subsamples of the data of
sample size 1000. Since IVQR requires a grid search over both coefficients and
quantiles, the resulting effect estimates may lead to quantile crossing
\citep{bassett1982}, which cannot appear using DIVE. The IVQR estimates show
larger variation, especially for more extreme quantiles (at 0.05 and 0.95),
compared to DIVE. 

\begin{figure}[t!]
    \centering
    \begin{subfigure}{0.49\textwidth}
    \includegraphics[width=\linewidth]{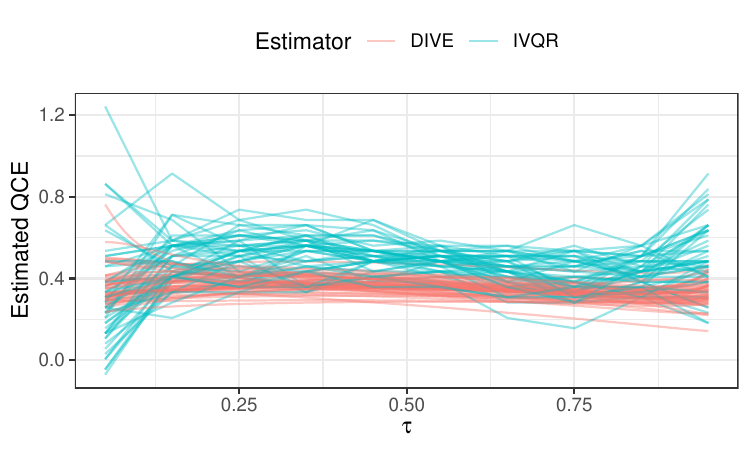}
    \subcaption{Schooling data}
    \end{subfigure}
    \begin{subfigure}{0.49\textwidth}
    \includegraphics[width=\linewidth]{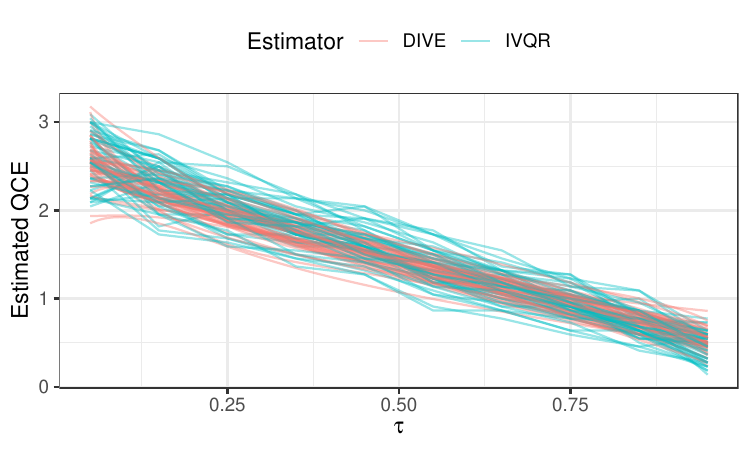}
    \subcaption{401(k) data}
    \end{subfigure}
    \caption{%
    Estimated quantile treatment effects for 50 subsamples of the schooling and
    401(k) datasets of size 1000 using DIVE and IVQR with linear quantile
    regression for $\tau \in (0.1, 0.9)$. 
    }\label{fig:qce}
\end{figure}

\subsection{Distributional causal effect of 401(k) participation on
net financial assets}\label{sec:401k}

We re-analyze the 401(k) dataset \citep{chernozhukov2004effects} using DIVE. We
model the distribution of log-transformed net total financial assets depending
on whether an individual participated in a 401(k) plan or not. The instrument is
the binary indicator of whether or not an individual was eligible to participate
in a 401(k) plan. Figure~\ref{fig:401k} summarizes the results. We use the same
settings for DIVE as in Section~\ref{sec:sim} but with $M=10$.

On the 401(k) data, DIVE estimates a larger distributional causal effect for all
wages compared to the observational conditional distribution, which again
translates in a larger estimated average causal effect. DIVE and 2SLS agree
closely in terms of direction and magnitude of the average causal effect. Also
in this application, the distributional causal effect estimates show that the
effect is larger for lower net total financial assets.
The interpretation of distributional causal effects pertains to populations and
cannot be interpreted as individual causal effect (this would be possible under
Assumption~\ref{asmp:crs}~\ref{asmp:crs:rinv} and, thus, access to the
unobservable joint distribution of individual responses under both treatments).

\begin{figure}[t!]
\centering
\includegraphics[width=0.99\textwidth]{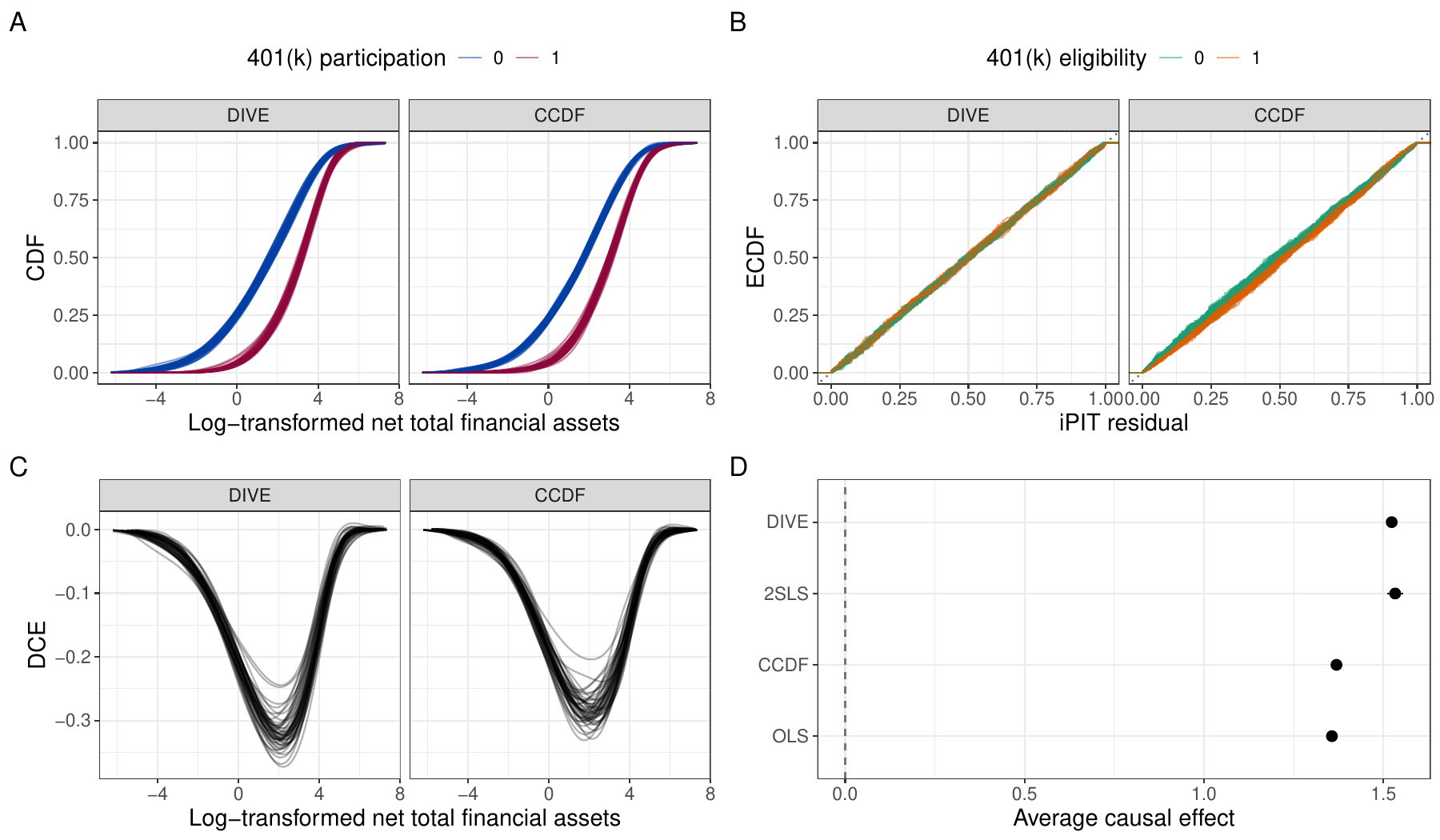}
\caption{%
Results for the 401(k) dataset from Section~\ref{sec:401k} for 50 subsamples of
the data of size 1000. \textsf{A}: Estimated CDFs of net total financial assets
by whether an individual participates in a 401(k) plan or not for DIVE and CCDF.
\textsf{B}: ECDFs of the empirical (interventional) PIT residuals stratified by
the instrument (whether a person is eligible to participate in a 401(k) plan).
\textsf{C}: Distributional causal effect estimated via DIVE and CCDF.
\textsf{D}: Comparison of the estimated average causal effect for DIVE, CCDF,
OLS and 2SLS (obtained via numerical integration for DIVE and CCDF). Standard
errors are based on 50 subsamples.
}\label{fig:401k}
\end{figure}

Figure~\ref{fig:qce} (b) shows the estimated QCE for both DIVE and IVQR. For
both methods, the estimated quantile causal effect decreases with increasing
quantiles. Similar to the schooling dataset, the variability of the IVQR
estimates is larger than that of DIVE over 50 subsamples of the dataset of size
1000.

\section{Discussion and outlook}\label{sec:discussion}

We propose a method for identifying and estimating distributional causal effects
using IV based on two properties of the iPIT residuals: Under the assumptions
stated in Theorem~\ref{thm:id}, only the interventional CDFs given the random
observed treatment $D$ produces iPIT residuals that are both uniformly
distributed and independent of the instrument. As shown both in simulations and
in real-data applications, the resulting estimator, DIVE, can be estimated
consistently and provides valuable information beyond the average causal effect.
Further, DIVE circumvents the issue of quantile crossing \citep{bassett1982} by
employing basis expansions instead of estimation over a grid of quantiles.

IVQR is the most similar method that applies in our setting (binary treatment
and absolutely continuous response). In contrast to our proposed IV-SCM, it is
based on the potential outcome model and uses a different identification
strategy. However, as shown in Proposition~\ref{prop:equiv_final} the two models
are related: For any IVQR model in which the structural quantile function is
identified, there exists a corresponding IV-SCM with the same observational and
interventional distributions (w.r.t.\ interventions on the treatment). Thus,
under appropriate relevance conditions, whenever IVQR is applicable, DIVE is
too. While available implementations of IVQR are limited to linear quantile
regression, they provide inference on the estimated quantile treatment effect
(for a fixed quantile) based on asymptotic normality established in
\citet{chern2007ivquantileinference}.

In this work, we did not consider adjusting for additional observed covariates.
A possible extension of DIVE would allow for this by estimating the conditional
interventional CDFs under treatment and control of the response given some
features, for which normalizing flows \citep{papamakarios2021} or other
nonparametric distributional regression methods could be used. 
When considering conditional IV, the choice between a distributional or
quantile-based estimand becomes important. This is because, in general, an
average quantile function and the corresponding inverse average CDF need not
agree. Both theoretical and computational advantages and disadvantages have been
discussed in \citet{leorato2015qdr}. Furthermore, it is an open problem whether
our identification strategy can be extended to continuous-valued treatments and
whether the control function approach proposed in
\citet{imbens2009identification} or the recently proposed copula-based approach
\citep{chernozhukov2024estimating} can be adapted to our setting. We further
discuss potential extensions of DIVE for censored and discrete responses in
Appendix~\ref{app:discrete}. Whether asymptotic confidence intervals,
analogously to IVQR, for the distributional causal effect at pre-specified
values of the response can be obtained is also left for future work.

\subsection*{Acknowledgments}
We thank Anton Rask Lundborg for fruitful discussions. LK was supported by the
Swiss National Science Foundation (Grant number 214457). NP was supported by a
research grant (0069071) from Novo Nordisk Fonden.

\vskip 0.2in
\bibliographystyle{abbrvnat}
\bibliography{ref}

\appendix
\renewcommand{\thesection}{\Alph{section}}
\counterwithin{figure}{section}
\renewcommand\thefigure{\thesection\arabic{figure}}
\counterwithin{table}{section}
\renewcommand\thetable{\thesection\arabic{table}}

\section{DIVE for discrete responses}\label{app:discrete}

In Setting~\ref{setting}, absolute continuity of $Y$ is a sufficient condition
to ensure that the iPIT residuals are uniformly distributed and indepdent of
$Z$. In Example~\ref{ex:non_binary} below, we present a setting with a binary
response $Y$ depending on a binary treatment $D$ without unmeasured confounding
and a binary instrument $Z$, in which $Y$ given $D$ is correctly specified by a
logistic regression. However, the iPIT residuals based on the correctly
specified model are neither uniform nor independent of $Z$. We leave extensions
to discrete and censored cases for future work.

\begin{example}
\label{ex:non_binary}
Consider the following IV-SCM
\begin{align}
    Z &\coloneqq 2 \1(N_Z \geq 0) - 1\\
    D &\coloneqq \1(Z + N_D \geq 0)\\
    Y &\coloneqq \1(-1 + D + N_Y \geq 0),
\end{align}
where $N_Z, N_D, N_Y$ are jointly independent and standard logistically
distributed. Here, for all $d \in \calD$ it holds that $F_d^*(0)=\mathbb{P}(-1
+ d + N_Y \leq 0)$ and $F_d^*(1)=1$. The iPIT residuals are therefore given by
\begin{align}
(1 - Y) F_D^*(Y) + Y = \begin{cases}
    \expit(1) & \mbox{if } D = 0, Y = 0,\\
    0.5 & \mbox{if } D = 1, Y = 0,\\
    1 & \mbox{if } D \in \{0, 1\}, Y = 1.
\end{cases}
\end{align}
This further implies that $(1-Y)F_D^*(Y)+Y\geq 0.5$ almost surely.
We can now compute
\begin{align}
    \Ex[Z((1 - Y)F_D^*(Y) + Y)] &= 
    \Prob(Z = 1) \Ex[(1 - Y)F_D^*(Y) + Y \given Z = 1] \\
    &= 0.5^2 \Ex[(1 - Y)F_D^*(Y) + Y \given D = 1, Z = 1] \ + \\
        &\quad\quad 0.5^2 \Ex[(1 - Y)F_D^*(Y) + Y \given D = 0, Z = 1] \\
    &= 0.5^2 \Ex[(1 - Y)F_D^*(Y) + Y \given D = 1] \ +\\
        &\quad\quad 0.5^2 \Ex[(1 - Y)F_D^*(Y) + Y \given D = 0] > 0.
\end{align}
Due to the discrete nature of the response, $F_D^*(Y)$ is also not uniformly
distributed. Further, contrary to the results obtained for
Settings~\ref{setting} and~\ref{setting:ivqr}, $F_D^*(Y) \not\indep Z$.
\end{example}

\section{Rank assumptions: Examples}\label{app:cexs}

\begin{example}[Violation of rank invariance and rank similarity]
\label{ex:nobreak}
Consider the following IV-SCM
\begin{align}
    H &\coloneqq N_H\\
    D &\coloneqq \1(N_D \geq 0.2 + 0.6\1(H > 0))\\
    Y &\coloneqq - DH^2 + (1 - D)H^2,
\end{align}
where $N_H \sim \ND(0,1)$ and $N_D \sim \UD(0,1)$. We now discuss the three
types of rank assumptions mentioned above and summarize the example in
Figure~\ref{fig:ex:nobreak}.
\begin{figure}[!ht]
\centering
\includegraphics[width=\textwidth]{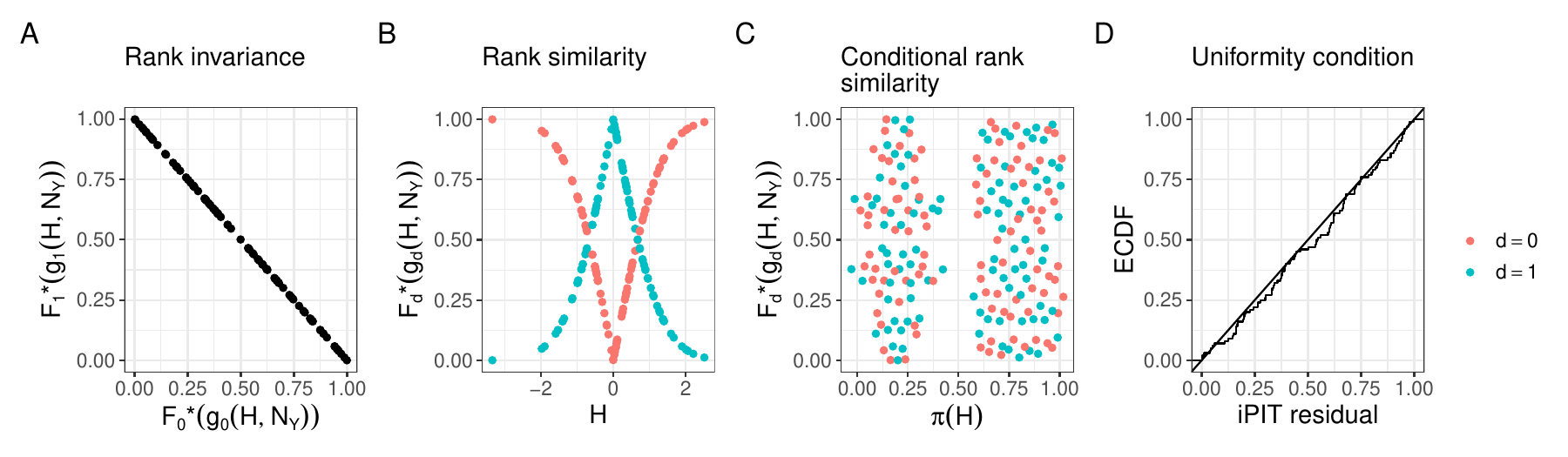}
\caption{%
Overview of the rank assumptions discussed in Example~\ref{ex:nobreak}.
A: Rank invariance does not hold.
B: Rank similarity does not hold.
C: Conditional rank similarity holds.
D: The iPIT residuals are uniform.
}\label{fig:ex:nobreak}
\end{figure}
Here, we have for all $d\in\calD$ and $h\in\RR$ that 
\[
g_d(h)=-dh^2+(1-d)h^2
\]
and for all $y\in\calY$ that 
\[
F_0^*(\ry)=F_{\chi^2_1}(\ry) \mbox{ and } F_1^*(\ry)= 1 - F_{\chi^2_1}(-\ry),
\]
where $F_{\chi^2_1}$ is the CDF of a chi-squared distribution with one degree of
freedom. We thus have 
\begin{align}\label{term:uneq}
F_0^*(g_0(H)) = 1 - F_1^*(g_1(H)) \mbox{ almost surely,}
\end{align}
and 
\[
F_D^*(g_D(H)) = F_0^*(g_0(H)) + D(F_1^*(g_1(H)) - F_0^*(g_0(H))).
\]
Now, define $W=F_1^*(g_1(H)) - F_0^*(g_0(H))$ and $M=F_0^*(g_0(H))$.
Then,
\begin{align}\label{term:ex:1}
E[WD+M\mid D=0, H]&=E[M\mid D=0, H]=M, \mbox{ and} \\
\label{term:ex:2}
E[WD+M\mid D=1, H]&=E[W+M\mid D=0, H]=W+M,
\end{align}
since both $W$ and $M$ are $H$-measurable. By \eqref{term:uneq}, we conclude
\eqref{term:ex:1} and \eqref{term:ex:2} are unequal almost surely and, thus,
rank invariance is violated. By Proposition~\ref{prop:rankasmp}, rank invariance
is violated too. We now turn to showing that conditional rank similarity is
fulfilled. The treatment propensity is given by 
\[
\pi(H) = 0.2 + 0.6\1(H > 0).
\]
In this specific example, it holds that
\begin{align}\label{ci:H2p}
H^2 \indep \pi(H).
\end{align}
To see this, observe that for all $h > 0$,
\[
\Prob(H^2 \leq h \given \pi(H) = 0.2) = 
\Prob(H^2 \leq h \given H \leq 0) = \Prob(H^2 \leq h) = F_{\chi^2_1}(h),
\]
and 
\[
\Prob(H^2 \leq h \given \pi(H) = 0.8) = 
\Prob(H^2 \leq h \given H > 0) = \Prob(H^2 \leq h) = F_{\chi^2_1}(h).
\]
Moreover, since $D\indep H\mid \pi(H)$ (Lemma~\ref{lem:balancing}), it also
holds that
\begin{align}\label{ci:H2}
D \indep H^2 \given \pi(H).   
\end{align}
Together, \eqref{ci:H2p} and \eqref{ci:H2} imply $(D, \pi(H)) \indep H^2$ by
contraction (Lemma~\ref{lem:ci}). Thus, for all $t\in(0,1)$,
\begin{align*}
\Prob(DW + M \leq t \given \pi(H), D=0) 
&= 
\Prob(M \leq t \given \pi(H), D=0) 
\\
&= 
\Prob(F_{\chi^2_1}(H^2) \leq t \given \pi(H), D = 0)
\\
&= \Prob(F_{\chi^2_1}(H^2) \leq t) = t,
\end{align*}
and
\begin{align*}
\Prob(DW + M \leq t \given \pi(H), D=1) 
&=
\Prob(W + M \leq t \given \pi(H), D=1) 
\\
&= \Prob(1 - F_{\chi^2_1}(H^2) \leq t \given \pi(H), D=1)
\\
&= \Prob(F_{\chi^2_1}(H^2) \leq t) = t,
\end{align*}
where we used that since $U \sim \UD(0, 1)$ it also holds that $1 - U \sim
\UD(0, 1)$. Thus, conditional rank invariance, $F_D^*(g_D(H)) \indep D \given
\pi(H)$, holds.
\end{example}

\begin{example}[Violation of all rank assumptions]\label{ex:break}
Consider the following IV-SCM
\begin{align*}
    H &\coloneqq N_H\\
    D &\coloneqq \1(N_D \geq H)\\
    Y &\coloneqq 2DH - H,
\end{align*}
where $N_H \sim \ND(0, 1)$ and $N_D$ follows a standard logistic distribution.
We now discuss the three types of rank assumptions mentioned above and summarize
the example in Figure~\ref{fig:ex:break}.
\begin{figure}[!ht]
\centering
\includegraphics[width=\textwidth]{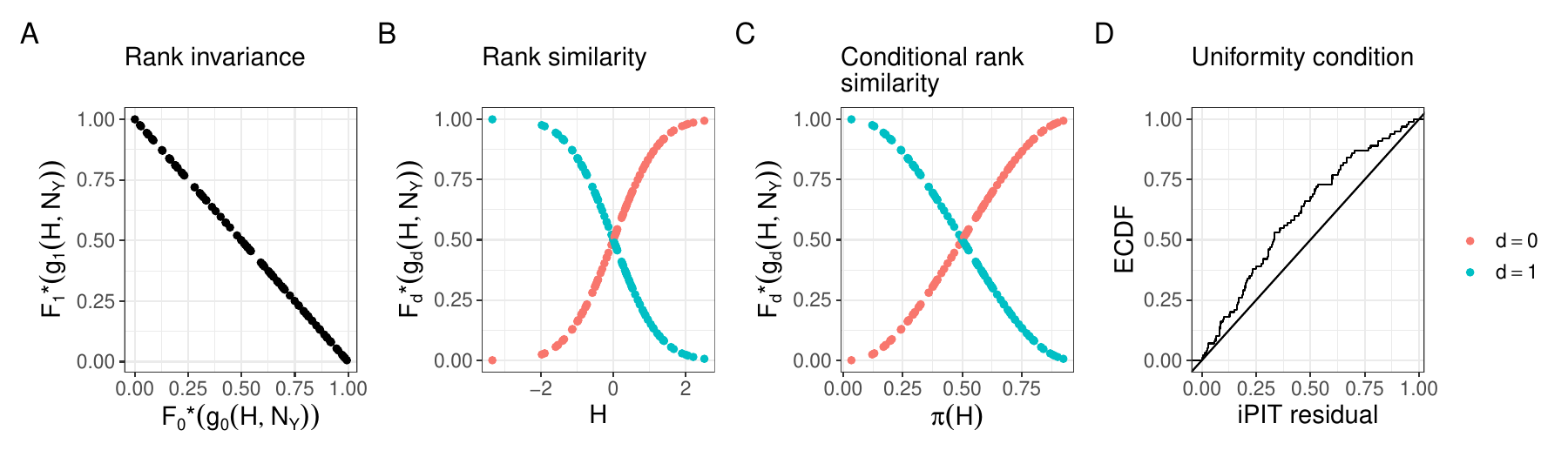}
\caption{%
Overview of the rank assumptions discussed in Example~\ref{ex:break}.
A: Rank invariance does not hold.
B: Rank similarity does not hold.
C: Conditional rank similarity does not hold.
D: The iPIT residuals are not uniform.
}\label{fig:ex:break}
\end{figure}
Here, we have for all $d\in\calD$ and $h\in\calH$ that $g_d(h)= 2dh-h$
and for all $y\in\calY$ that $F_0^*(\ry)=\Phi(-h)=1-\Phi(h)$ and $F_1^*(\ry) =
\Phi(h)$.
Conditional rank similarity is not fulfilled, because we have $\pi(H) =
\expit(H)$ and thus $\Prob(F^*_0(g_0(H)) \leq t \given \pi(H)) =
\Prob(F^*_0(g_0(H)) \leq t \given H) = \Prob[\Phi(-H) \leq t \given H) = 1 -
\Prob(\Phi(H) \leq t \given H)$ and $\Prob(F^*_1(g_1(H)) \leq t \given H) =
\Prob(\Phi(H) \leq t \given H)$ 
By the implications in Proposition~\ref{prop:rankasmp}, rank similarity and rank
invariance are violated as well. 
\end{example}

\section{Additional numerical results}\label{app:additional}

\subsection{Simulation scenarios}\label{app:simscen}

Figure~\ref{fig:simscen} depicts the observational and interventional CDFs under
treatments $0$ and $1$ for the simulation scenarios described in
Section~\ref{sec:sim}.

\begin{figure}[t!]
\centering
\includegraphics[width=0.9\textwidth]{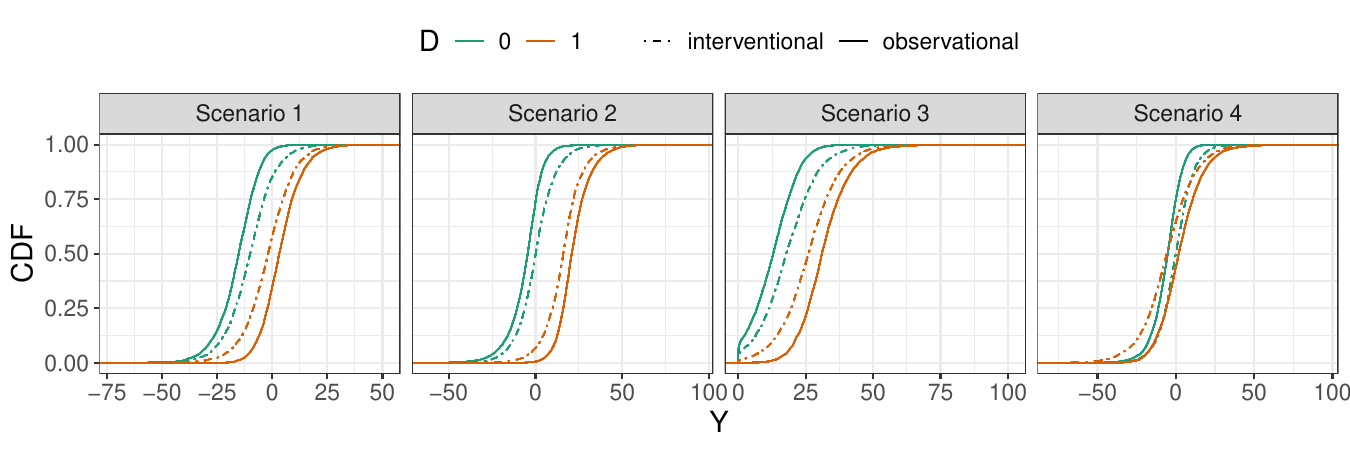}
\caption{%
Observational and interventional CDFs for the simulation scenarios from
Section~\ref{sec:sim}. In the first scenario, the interventional CDFs under both
treatments are logistic. In the second, the interventional CDFs under both
treatments are an average of logistic CDFs. In the third scenario, a non-linear
transformation bounds the outcome from below at zero. The fourth scenario is
considerably more challenging as the observational CDFs cross at lower values of
the response whereas the interventional CDFs cross at larger values of the
response.
}
\label{fig:simscen}
\end{figure}

\subsection{Combining uniformity and independence losses}\label{app:sigma}

Figure~\ref{fig:app:max} shows the results when using $\aggr(a,b) \coloneqq
\max\{a, b\}$ in the DIV loss (Definition~\ref{def:divl}) instead of the
weighted sum used in the main text. There is no indication that the choice of
$\aggr$ changes the interpretation of the results.

\begin{figure}[!ht]
\centering
\begin{subfigure}{0.99\textwidth}
\includegraphics[width=0.99\textwidth]{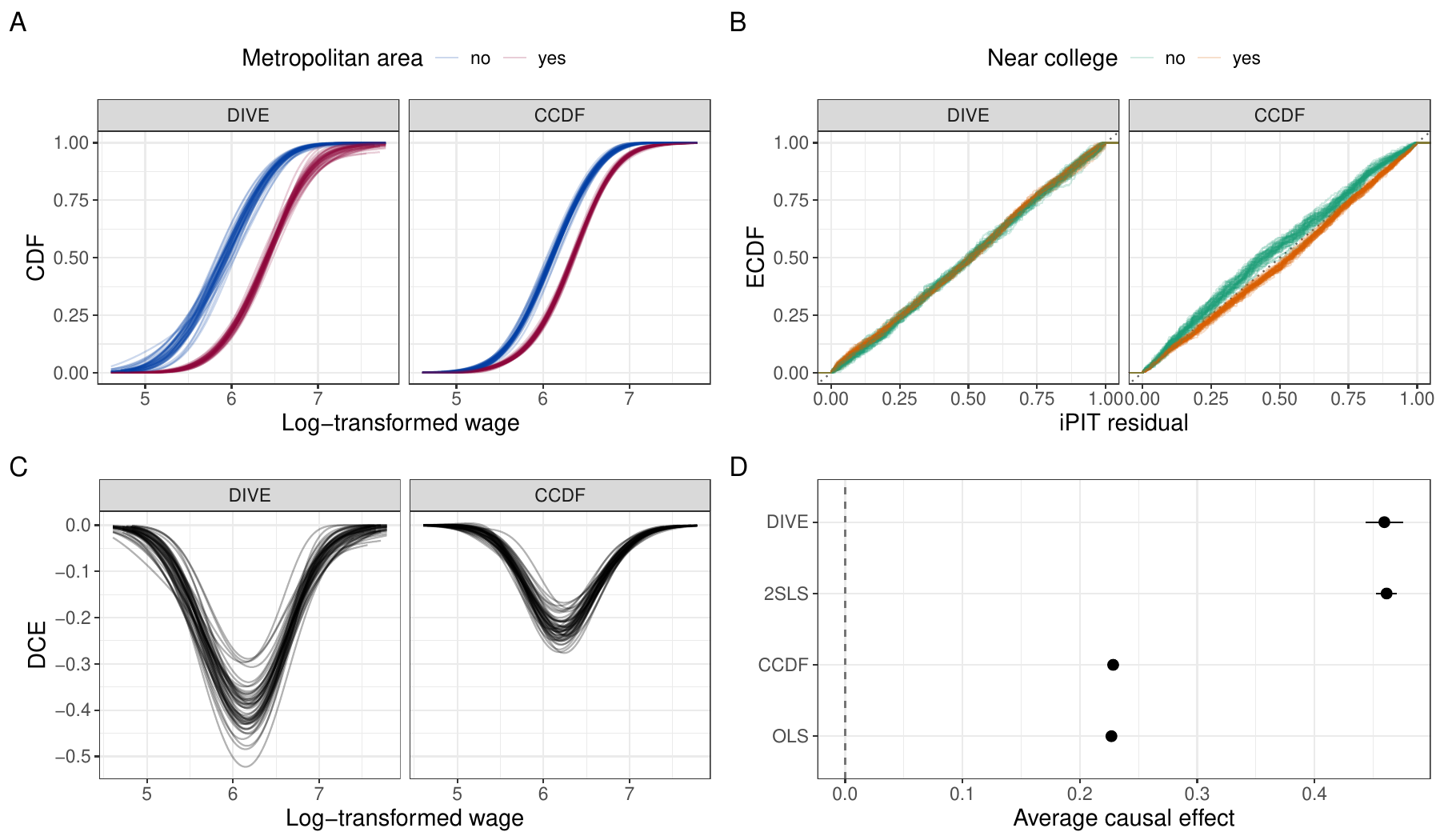}
\subcaption{Schooling data}
\end{subfigure}
\begin{subfigure}{0.99\textwidth}
\includegraphics[width=0.99\textwidth]{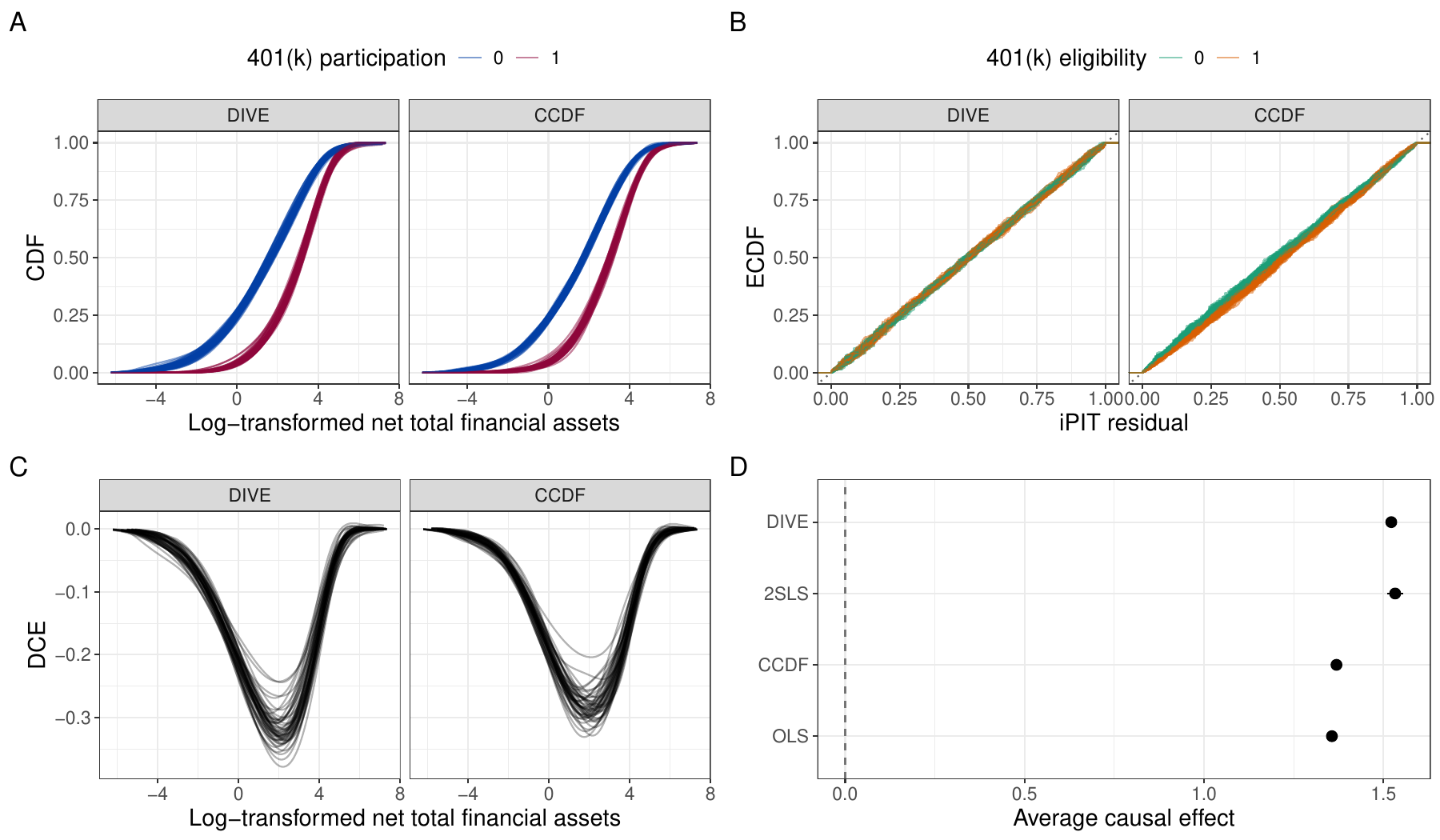}
\subcaption{401(k) data}
\end{subfigure}
\caption{%
Results for the schooling return dataset from Section~\ref{sec:schooling} 
and 401(k) dataset from Section~\ref{sec:401k} for 50 subsamples of the
data of size 1000 using $\aggr(a,b) \coloneqq \max\{a, b\}$.
}
\label{fig:app:max}
\end{figure}

\clearpage

\section{Proofs}\label{app:proofs}

\subsection{Proof of Theorem~\ref{thm:id}} \label{proof:thm:id}

\begin{proof}
We begin with the backward direction. Using the adjustment formula
\citep[e.g.,][]{pearl2009causality}, for all $d\in\calD$, the interventional CDF
under treatment $d$ of $Y$ can be expressed, for all $y\in\calY$, as
\begin{align}\label{eq:icdf}
F^*_d(y) = \Ex[\1(g_d(H, N_Y) \leq y)].
\end{align}

We fist show the case where Assumptions~\ref{asmp:pcu}
and~\ref{asmp:crs}~\ref{asmp:crs:crs} are satisfied. For all $t \in (0, 1)$, we
have
\begin{align}
&\Prob(F_D^*(Y) \leq t) \\
\label{from:def}
&= \Ex[\1(F^*_D(Y) \leq t)] \\
\label{from:scm}
&= \Ex[\1(F^*_D(g_D(H, N_Y)) \leq t)] \\
\label{from:tower}
&= \Ex[\Ex[D \1(F^*_1(g_1(H, N_Y)) \leq t) + 
    (1 - D) \1(F^*_0(g_0(H, N_Y)) \given H, N_Y]] \\
\label{from:hmeasurable}
&= \Ex[\Ex[D \given H, N_Y] \1(F^*_1(g_1(H, N_Y)) \leq t) + 
    \Ex[(1 - D) \given H, N_Y]\1(F^*_0(g_0(H, N_Y)) \leq t)] \\
\label{from:qmeasurable}
&= \Ex[\Ex[\pi(H)\1(F^*_1(g_1(H, N_Y)) \leq t) + 
    (1 - \pi(H))\1(F^*_0(g_0(H, N_Y)) \leq t) \given \pi(H)]] \\
\label{from:qmeasurable2}
&= \Ex[\pi(H)\Ex[\1(F^*_1(g_1(H, N_Y)) \leq t)\given \pi(H)] + 
    (1 - \pi(H))\Ex[\1(F^*_0(g_0(H, N_Y)) \leq t) \given \pi(H)]] \\
\label{from:asmp}
&= \Ex[\pi(H)\Ex[\1(F^*_0(g_0(H, N_Y)) \leq t)\given \pi(H)] + 
    (1 - \pi(H))\Ex[\1(F^*_0(g_0(H, N_Y)) \leq t) \given \pi(H)]] \\
\label{from:rearrange}
&= 
\Prob(F^*_0(g_0(H, N_Y)) \leq t)
\\ 
\label{from:combine}
&= 
\Prob(g_0(H, N_Y) \leq (F_0^*)^{-1}(t))
\\
&= t. \label{from:cdf}
\end{align}
For these equalities, we used the following arguments: \eqref{from:def} follows
by expressing the probability as an expectation; \eqref{from:scm} follows from
substituting the structural equation of $Y$; \eqref{from:tower} follows by the
law of total expectation and since $D$ is binary; \eqref{from:hmeasurable}
follows since the arguments in the indicator function are $(H,N_Y)$-measurable;
\eqref{from:qmeasurable} follows by the law of total expectation;
\eqref{from:qmeasurable2} follows since $\pi(H)$ is $\pi(H)$-measurable;
\eqref{from:asmp} follows 
because, for all $d \in \calD$, $\Ex[F^*_D(Y) \given \pi(H)] = \Ex[F^*_D(Y)
\given \pi(H), D = d] = \Ex[F^*_d(Y) \given \pi(H), D = d]$
by Assumption~\ref{asmp:crs}~\ref{asmp:crs:crs};
\eqref{from:rearrange} follows by direct computation;
\eqref{from:combine} follows by taking the generalized
inverse of $F_{0}^*$ (which exists since $g_0(H, N_Y)$ is assumed to be
absolutely continuous; see Setting~\ref{setting})
on both sides of the inequality. We can hence conclude that $F_D^*(Y)$ is
uniformly distributed. 

To establish that $F_D^*(Y)$ is independent of $Z$, consider the interventional
distribution of $Y$ conditional on the instrument $Z$. For all $t \in (0, 1)$,
we have
\begin{align}
&\Prob(F_D^*(Y) \leq t \given Z) \\
\label{f:cdef}
&= \Ex[\1(F^*_D(Y) \leq t) \given Z] \\
\label{f:ctower}
&= \Ex[\Ex[\1(F^*_D(Y) \leq t) \given D, \pi(H), Z] \given Z] \\
\label{f:dbin}
&= \Ex[(1-D)\Ex[\1(F^*_0(Y) \leq t) \given D, \pi(H), Z] + 
    D\Ex[\1(F^*_1(Y) \leq t) \given D, \pi(H), Z] \given Z] \\
\label{f:asmp1}
&= \Ex[(1-D)\Ex[\1(F^*_0(Y) \leq t) \given D, \pi(H)] + 
    D\Ex[\1(F^*_1(Y) \leq t) \given D, \pi(H)] \given Z] \\
\label{f:comp}
&= \Ex[\Ex[\1(F^*_D(Y) \leq t) \given D, \pi(H)] \given Z] \\
\label{f:asmp2}
&= \Ex[\Ex[\1(F^*_D(Y) \leq t) \given \pi(H)] \given Z] \\
\label{f:HindepZ}
&= \Ex[\Ex[\1(F^*_D(Y) \leq t) \given \pi(H)]] \\
&= \Ex[\1(F^*_D(Y) \leq t)] = t
\end{align}
where \eqref{f:cdef} follows again by expressing the probability as an
expectation; \eqref{f:ctower} follows by the law of total expectation;
\eqref{f:dbin} follows since $D$ is binary; \eqref{f:asmp1} follows from
Assumption~\ref{asmp:pcu}; \eqref{f:comp} follows from direct computation;
\eqref{f:asmp2} follows from Assumption~\ref{asmp:crs}~\ref{asmp:crs:crs}; and
\eqref{f:HindepZ} follows by $N_H \indep N_Z$ (Setting~\ref{setting}). We thus
conclude $F_D^*(Y) \indep Z$, which completes the backward direction.

Next we show the same result when assuming Assumption~\ref{asmp:crs}
\ref{asmp:crs:rsim} instead of Assumptions~\ref{asmp:pcu}
and~\ref{asmp:crs}~\ref{asmp:crs:crs}. A similar line of argument as
\eqref{from:def}--\eqref{from:combine} to show uniformity of $F_D^*(Y)$ holds:
For all $t \in (0, 1)$, we have
\begin{align}
&\Prob(F_D^*(Y) \leq t) \\
&= \Ex[\Ex[D\1(F^*_1(g_1(H, N_Y)) \leq t) + 
    (1 - D)\1(F^*_0(g_0(H, N_Y)) \leq t) \given H]] \\
&= \Ex[\pi(H)\Ex[\1(F^*_1(g_1(H, N_Y)) \leq t)\given H] + 
    (1 - \pi(H))\Ex[\1(F^*_0(g_0(H, N_Y)) \leq t) \given H]] \\
&= \Ex[\pi(H)\Ex[\1(F^*_0(g_0(H, N_Y)) \leq t)\given H] + 
    (1 - \pi(H))\Ex[\1(F^*_0(g_0(H, N_Y)) \leq t) \given H]] \\
&= \Prob(F^*_0(g_0(H, N_Y)) \leq t) =  \Prob(g_0(H, N_Y) \leq (F_0^*)^{-1}(t)) = t.
\end{align}
To establish independence of $Z$ and $F_D^*(Y)$ under Assumption~\ref{asmp:crs}
\ref{asmp:crs:rsim}, we first obtain, for all $t\in[0,1]$, that
\begin{align}
&\Prob(F_D^*(Y) \leq t \given Z) \\
&= \Ex[\1(F^*_D(Y) \leq t) \given Z] \\
&= \Ex[\1(F^*_D(g_D(H, N_Y)) \leq t) \given Z] \\
&= \Ex[D \1(F^*_1(g_1(H, N_Y)) \leq t) + 
    (1 - D) \1(F^*_0(g_0(H, N_Y)) \given Z]\\
\label{ff:lote}
&= \Ex[D\Ex[\1(F_1^*(g_1(H, N_Y)) \leq t) \given D, Z, H] +
(1-D)\Ex[\1(F_0^*(g_0(H, N_Y)) \leq t) \given D, Z, H]\given Z] \\
\label{ff:lemcis1}
&= \Ex[D\Ex[\1(F_1^*(g_1(H, N_Y)) \leq t) \given H] + (1-D)\Ex[\1(F_0^*(g_0(H,
N_Y)) \leq t) \given H]\given Z],
\end{align}
where \eqref{ff:lote} follows by the law of total expectation and
\eqref{ff:lemcis1} follows by $(Z, D) \indep (H, N_Y) \given H$
(Lemma~\ref{lem:cis1}). Then, combining this with Assumption~\ref{asmp:crs}
\ref{asmp:crs:rsim}, we get that
\begin{align}
&\Prob(F_D^*(Y) \leq t \given Z) \\
&= \Ex[D\Ex[\1(F_0^*(g_0(H, N_Y)) \leq t) \given H] \given Z] +
   \Ex[(1-D)\Ex[\1(F_0^*(g_0(H, N_Y)) \leq t) \given H] \given Z]\\
&= \Ex[\Ex[\1(F_0^*(g_0(H, N_Y)) \leq t) \given H] \given Z]\\
&= \Ex[\Ex[\1(F_0^*(g_0(H, N_Y)) \leq t)\mid H]] = t,
\end{align}
where for the third equality we used that
$N_H \indep N_Z$ (Setting~\ref{setting}).
This completes the proof of the backward direction given
Assumption~\ref{asmp:crs} \ref{asmp:crs:rsim} instead of
Assumptions~\ref{asmp:pcu} and~\ref{asmp:crs}~\ref{asmp:crs:crs}.

Next, consider the forward direction and fix arbitrary $F_0, F_1 \in \calF$ such
that $F_D(Y)\indep Z$ and $F_D(Y)\sim\operatorname{Uniform}(0,1)$. We use the
structural equation of $Y$ to get
\begin{align}
    F_D(Y)=D F_1(Y) + (1 - D) F_0(Y) = F_0(g_0(H, N_Y)) + 
    D(F_1(g_1(H, N_Y)) - F_0(g_0(H, N_Y))).
\end{align}
Since $F_D(Y)\indep Z$ we have by Assumption~\ref{asmp:relevance}, that
$\Prob_{(H, N_Y)}$-almost surely,
\begin{align}
    F_1(g_1(H, N_Y)) = F_0(g_0(H, N_Y)).
\end{align}
By the above, the premise $F_D(Y) \sim \UD(0,1)$ implies that $F_0(g_0(H, N_Y))
\sim \UD(0, 1)$. For all $d \in\calD$, $Y$ is equal to $g_d(H, N_Y)$ under
$\pdo(D = d)$, which, together with the uniqueness of the CDF of absolutely
continuous random variables (Lemma~\ref{lem:cdfunique}), implies $F_0 = F_0^*$.
By the same argument, we have $F_1(g_1(H, N_Y)) \sim \UD(0, 1)$ and $F_1 =
F_1^*$, which concludes the proof.
\end{proof}

\subsection{Proof of Proposition~\ref{prop:rankasmp}}
\label{proof:prop:rankasmp}

\begin{proof}
Let $\phi:\mathbb{R}\rightarrow\mathbb{R}$ be an arbitrary bounded continuous
function. Then, it holds that
\begin{align*}
    &\mathbb{E}[\phi(F^*_D(g_D(H, N_Y)))\given H]\\
    &\quad=\mathbb{E}[D\phi(F^*_1(g_1(H, N_Y)))\given H] +
    \mathbb{E}[(1-D)\phi(F^*_0(g_0(H, N_Y)))\given H]\\
    &\quad=\mathbb{E}[D\given H]\mathbb{E}[\phi(F^*_1(g_1(H, N_Y)))\given H] +
    \mathbb{E}[(1-D)\given H]\mathbb{E}[\phi(F^*_0(g_0(H, N_Y)))\given H]\\
    &\quad=\pi(H)\mathbb{E}[\phi(F^*_1(g_1(H, N_Y)))\given H] +
    (1-\pi(H))\mathbb{E}[\phi(F^*_0(g_0(H, N_Y)))\given H],
\end{align*}
where for the second equality we used that $D\indep (H, N_Y)\given H$ (which
follows directly from Lemma~\ref{lem:cis1} using decomposition from
Lemma~\ref{lem:ci}). Similarly, we get
\begin{align*}
    &\mathbb{E}[\phi(F^*_D(g_D(H, N_Y)))\given D, H]\\
    &\quad=\mathbb{E}[\mathbb{E}[\phi(F^*_D(g_D(H, N_Y)))\given D, H, N_Y]\given
    D, H]\\
    &\quad=\mathbb{E}[\mathbb{E}[\phi(F^*_D(g_D(H, N_Y)))\given H, N_Y]\given D,
    H]\\
    &\quad=\mathbb{E}[\mathbb{E}[D\phi(F^*_1(g_1(H, N_Y)))\given H, N_Y] +
    \mathbb{E}[(1-D)\phi(F^*_0(g_0(H, N_Y)))\given H, N_Y]\given D, H]\\
    &\quad=\mathbb{E}[\mathbb{E}[D\given H, N_Y] \phi(F^*_1(g_1(H, N_Y))) +
    \mathbb{E}[(1-D)\given H, N_Y]\phi(F^*_0(g_0(H, N_Y)))\given D, H]\\
    &\quad=\mathbb{E}[\mathbb{E}[D\given H] \phi(F^*_1(g_1(H, N_Y))) +
    \mathbb{E}[(1-D)\given H]\phi(F^*_0(g_0(H, N_Y)))\given D, H]\\
    &\quad=\mathbb{E}[\pi(H) \phi(F^*_1(g_1(H, N_Y))) +
    (1-\pi(H))\phi(F^*_0(g_0(H, N_Y)))\given D, H]\\
    &\quad=\pi(H)\mathbb{E}[ \phi(F^*_1(g_1(H, N_Y)))\given H, D] +
    (1-\pi(H))\mathbb{E}[\phi(F^*_0(g_0(H, N_Y)))\given D, H]\\
    &\quad=\pi(H)\mathbb{E}[\phi(F^*_1(g_1(H, N_Y)))\given H] +
    (1-\pi(H))\mathbb{E}[\phi(F^*_0(g_0(H, N_Y)))\given H],
\end{align*}
where we used Assumption~\ref{asmp:crs} \ref{asmp:crs:rinv} for the second
equality and $D \indep (N_Y, H) \given H$ (as above) for the last equality.
Since $\phi$ was arbitrary and using the definition of conditional independence
(Definition~\ref{def:condind}), we conclude that Assumption~\ref{asmp:crs}
\ref{asmp:crs:rinv} implies Assumption~\ref{asmp:crs} \ref{asmp:crs:rsim}.

To show the second implication, we use
\begin{align}
F_D^*(g_D(H, N_Y)) \indep D \given H \iff
F_D^*(g_D(H, N_Y)) \indep D \given H, \pi(H)
\end{align}
which holds by the definition of conditional independence (since the sigma
algebras of $H$ and $(H, \pi(H))$ are identical as $\pi(H)$ is $H$-measurable),
together with $D \indep H \given \pi(H)$ (which holds by
Lemma~\ref{lem:balancing}),
to conclude
\[
F_D^*(g_D(H, N_Y)) \indep D \given \pi(H)
\]
by contraction (Lemma~\ref{lem:ci}). Counterexamples showing that the reverse
direction of the first and second implications in
Proposition~\ref{prop:rankasmp} are given in Examples~\ref{ex:holds}
and~\ref{ex:nobreak}, respectively.
\end{proof}

\subsection{Proof of Proposition~\ref{prop:equiv_estimand}}
\label{proof:prop:equiv_estimand}

\begin{proof}
We begin by fully specifying the parametrization of the IVQR models defined in
Setting~\ref{setting:ivqr}.

Let $\calP$ denote the set of all distributions on $[0, 1] \times [0, 1]\times
\RR^{d_V} \times \calZ$ and define the sets $\mathcal{D} \coloneqq \{\delta :
\calZ \times \RR^{d_V} \to \{0,1\} \given \delta \mbox{ measurable}\}$ and
$\calQ \coloneqq \{q : [0, 1] \to \calY \given q \mbox{ measurable and strictly
increasing}\}$. All IVQR models in Setting~\ref{setting:ivqr} can then be
parametrized by the following set
\begin{align*}
\ivqr\coloneqq\Big\{ (\nu, \delta, q_0, q_1)\in\mathcal{P}\times\mathcal{D}
  \times\mathcal{Q}\times\mathcal{Q}\,\big|\, (U(0), U(1),V, Z)\sim\nu:\,
\text{satisfying Setting~\ref{setting:ivqr}} \Big\}.
\end{align*}

Similarly, we now parametrize a subset of the IV-SCMs in Setting~\ref{setting}
for $\calH\coloneqq\mathbb{R}^{d_{V}}$ and $N_D=0$.  Let $\calP'$ denote the set
of all distributions on $\RR\times [0,1]^2\times\RR^{d_{V}} \times \RR^{d_Z}$
and define the sets $\mathcal{D}' \coloneqq \{f : \calZ \times \calH \times \RR
\to \{0, 1\} \given f \mbox{ measurable}\}$ and $\calQ' \coloneqq \{g : \calH
\times [0,1]^2 \to \calY \given g \mbox{ measurable}\}$. Then, we define the
set,
\begin{align*}
\ivscm\coloneqq\Big\{ &(\mu, f, g_0, g_1)\in \mathcal{P}'\times\mathcal{D}'
  \times\mathcal{Q}'\times\mathcal{Q}'\,\big|\, (N_D, N_Y, N_H, N_Z)\sim
  \mu,\,\\ &\quad \text{with } N_D = 0 \mbox{ a.s.} \text{ and }(N_Z, N_D)\indep
  N_Y\given N_H \mbox{ and } N_H \indep N_Z
  \Big\}.
\end{align*}
We now explicitly define a map
\begin{equation*}
\calG : \ivqr \to \ivscm
\end{equation*}
and then show that it satisfies the desired properties. To this end, we define
for all $M=(\nu, \delta, q_0, q_1)\in\ivqr$,
\begin{equation*}
\mathcal{G}(M)\coloneqq (\mu^M, f^M, g_0^M, g_1^M),
\end{equation*}
where
\begin{itemize}
\item $\mu^M\in\mathcal{P}'$ is defined by
\begin{equation*}
  \mu^M\coloneqq (\Delta_0 \otimes \nu_{U(0), U(1), V, Z}, f, g_0, g_1),
\end{equation*}
where $\Delta_0$ the Dirac measure at zero on $\RR$.
\item $f^M\in\mathcal{D}'$ for all
$(z,h,n)\in\mathcal{Z}\times\mathcal{H}\times\mathbb{R}$ is defined by
\begin{equation*}
  f^M(z, h, n)\coloneqq\delta(z, h).
\end{equation*}
\item $g_0^M\in\mathcal{Q}'$ and $g_1^M\in\mathcal{Q}'$ for all
$(h, n)\in\mathcal{H}\times [0,1]^2$ 
are defined by
\begin{align}
  g_0(h, n)\coloneqq q_0(n_1)
  \quad\text{and}\quad
  g_1(h, n)\coloneqq q_1(n_2).
\end{align}
\end{itemize}
By construction of $\mu^M$ it holds for $(N_D,N_Y,N_H,N_Z)\sim \mu^M$ that
$N_D=0$. Moreover, since by assumption $Z\indep (U(0), U(1),V)$
for $(U(0), U(1),V, Z)\sim \nu$ we also get (using Lemma~\ref{lem:ci}) that
$(N_Z,N_D)\indep N_Y\given N_H$ and $N_H \indep N_Z$ for $(N_D,N_Y,N_H,N_Z)\sim
\mu^M$, hence $\mathcal{G}(M)\in\ivscm$ as desired.

We are now ready to prove the main result. To this end, fix $M=(\nu, \delta,
q_0, q_1)\in\ivqr$. Then, this induces a unique distribution $P^M$ over
\begin{equation*}
(Z, D, U(0), U(1), Y(0), Y(1))
\end{equation*}
by first sampling $(U(0),U(1),V,Z)\sim \nu$ and then setting $D=\delta(Z, V)$,
$Y(0)=q_0(U(0))$ and $Y(1)=q_1(U(1))$. By construction $P^M$ satisfies the
conditions of an IVQR model as in Setting~\ref{setting:ivqr}. We now prove (1)
the equivalence of the observational distributions, (2) equivalence of the
interventional distributions and (3) that $\mathcal{G}(M)$ satisfies rank
similarity (Assumption~\ref{asmp:crs}~\ref{asmp:crs:rsim}).

\textit{Step 1: Equivalence of observational distribution:}\\
Based on the construction of $P^M$ it holds for $(U(0),U(1),V,Z)\sim \nu$ that
\begin{equation}
\label{eq:obs_ivqr}
\left(Z, \delta(Z, V), q_{\delta(Z, V)}(U(\delta(Z,
  V)))\right)\sim P^M_{Z,D,Y(D)}.
\end{equation}
For the IV-SCM $\mathcal{G}(M)$ we can use the definition of the function $f^M$
to express the structural equation of $D$ in terms of the noise variables
$(N_D,N_Y,N_H,N_Z)$ as follows
\begin{equation}
\label{eq:Dstructural_eq}
D=f^M(Z, H, N_D)=\delta(Z, H)=\delta(N_Z, N_H).
\end{equation}
Similarly, using this expansion and the definitions of $g_0^M$ and $g_1^M$ we
get for the structural equation of $Y$ that
\begin{equation}
\label{eq:Ystructural_eq}
Y=g_D(H, N_Y)=q_D((N_Y)_{1+D})=q_{\delta(N_Z, N_H)}((N_Y)_{1+\delta(N_Z, N_H)}).
\end{equation}
Together with $Z=N_Z$ it therefore holds for $(N_D, N_Y, N_H, N_Z)\sim\mu^M$
that
\begin{equation}
\label{eq:obs_ivscm}
(N_Z, \delta(N_Z, N_H),q_{\delta(N_Z, N_H)}((N_Y)_{1+\delta(N_Z, N_H)}))\sim
P^{\mathcal{G}(M)}_{Z,D,Y}. 
\end{equation}
By definition of $\mu^M$ it holds for $(N_D,N_Y, N_H, N_Z)\sim\mu^M$ that
$(N_Y,N_H, N_Z)\sim\nu$. Hence, \eqref{eq:obs_ivqr} and \eqref{eq:obs_ivscm}
together imply that
\begin{equation*}
P^M_{Z,D,Y(D)}=P^{\mathcal{G}(M)}_{Z,D,Y}.
\end{equation*}

\textit{Step 2: Equivalence of interventional distributions:}\\
Again by the construction of $P^M$ we get for $(U(0), U(1),V,Z)\sim\nu$ that
\begin{equation}
\label{eq:int_ivqr}
q_0(U(0))\sim P^M_{Y(0)} \quad\text{and}\quad q_1(U(1))\sim P^M_{Y(1)}.
\end{equation}
Furthermore, performing $\pdo(D=0)$ (respectively, $\pdo(D=1)$) and then using
the definition of $g_0^M$ and $g_1^M$ together with the intervened structural
equations show for $(N_D,N_Y,N_H,N_Z)\sim\mu^M$ that
\begin{equation} \label{eq:int_ivscm}
q_0((N_Y)_{1})\sim P^{\mathcal{G}(M);\pdo(D=0)}_Y \quad\text{and}\quad
q_1((N_Y)_{2})\sim P^{\mathcal{G}(M);\pdo(D=1)}_Y.
\end{equation}
Hence, using again the definition of $\mu^M$, \eqref{eq:int_ivqr} and
\eqref{eq:int_ivscm} imply that
\begin{equation*}
P^{M}_{Y(0)}=P^{\mathcal{G}(M);\pdo(D=0)}_Y \quad\text{and}\quad
P^{M}_{Y(1)}=P^{\mathcal{G}(M);\pdo(D=1)}_Y.
\end{equation*}

\textit{Step 2: Rank similarity:}\\
By \eqref{eq:int_ivscm} and since $(N_Y)_{1}$ and $(N_Y)_{2}$ are both
distributed uniformly on $[0,1]$ (using the definition of $\mu^M$, assumptions
on $\nu$ and that $q_0$ and $q_1$ are strictly increasing) it holds for all
$d\in\{0,1\}$ that $F^*_d=q_d^{-1}$. Moreover, using \eqref{eq:Dstructural_eq}
and \eqref{eq:Ystructural_eq} it follows for $(N_D,N_Y,N_H,N_Z)\sim\mu^M$ that
\begin{equation} \label{eq:noise_expression_ranksim}
F^*_D(g_D(H, N_Y))=(N_Y)_{1+\delta(N_Z, N_H)}, \quad D=\delta(N_Z, N_H)
\quad\text{and}\quad H=N_H.
\end{equation}
Moreover, by the connection between $\mu^M$ and $\nu$ it holds that
\begin{equation} \label{eq:rank_sim_eq1}
(N_D,N_Y,N_H,N_Z)\sim\mu^M:\quad (N_Y)_{1+\delta(N_Z, N_H)}\indep \delta(N_Z,
N_H)\given N_H \end{equation}
if and only if
\begin{equation} \label{eq:rank_sim_eq2}
(U(0),U(1),V,Z)\sim\nu:\quad U(\delta(Z, V))\indep \delta(Z, V)\given V.
\end{equation}
Now rank similarity in the IVQR model (i.e., $U(D)\indep D \given V, Z$)
together with the independence assumption $(U(0), U(1))\indep Z\given V$ imply
that \eqref{eq:rank_sim_eq2} and therefore also \eqref{eq:rank_sim_eq1} is true.
Hence, using \eqref{eq:noise_expression_ranksim}, it holds that $\mathcal{G}(M)$
satisfies Assumption~\ref{asmp:crs}~\ref{asmp:crs:rsim}.

This completes the proof of Proposition~\ref{prop:equiv_final}.
\end{proof}

\subsection{Proof of Proposition~\ref{prop:consistency}}\label{proof:consistency}

\begin{proof}\sloppy
Fix an arbitrary $\lambda > 0$. Define $s : \calY \to \RR$ for all $y\in\calY$
by $s(y)\coloneqq \frac{y - L}{U - L}$, $\Theta \coloneqq \Theta_{M} \times
\Theta_{M}$, $\parm^* \coloneqq \left((\parm^*_0)^\top,
(\parm^*_1)^\top\right)^\top$, and, for all $(d, y)\in\calD\times\calY$ define
\[
a_d(y)\coloneqq \left((1-d)a_{M}(s(y))^\top, da_{M}(s(y))^\top\right)^\top. 
\]
Furthermore, for all $(d, y) \in \calD \times \calY$ and $\parm \in \Theta$,
define 
\[
F_d(y; \parm) \coloneqq \Phi\left(a_d(y)^\top\parm\right),
\]
and the following sample and population quantities:
\begin{align}
\label{term:hsicn}
{H}_n(\parm) &\coloneqq \hat{\operatorname{HSIC}}((F_{D_i}(Y_i; \parm))_{i=1}^n,
(Z_i)_{i=1}^n; k_R, k_Z), \\
\label{term:hsicp}
H(\parm) &\coloneqq \operatorname{HSIC}(F_D(Y;\parm), Z; k_R, k_Z), \\
\label{term:unifn}
U_n(\parm) &\coloneqq
\int_0^1 \left(\frac{1}{n}
\sum_{i=1}^n \1(F_{D_i}(Y_i;\parm) \leq u) - u \right)^2 \dd u
\\
\label{term:unifp}
U(\parm) &\coloneqq \int_0^1 (\Prob(F_D(Y;\parm) \leq u) - u)^2 \dd u, \\
\label{term:lossn}
\ell_{\lambda,n}(\parm) &\coloneqq 
\aggr(\lambda H_n(\parm), U_n(\parm)), \\
\label{term:lossp}
\ell_\lambda(\parm) &\coloneqq \aggr(\lambda H(\parm), U(\parm)),
\end{align}
where the population HSIC is defined in Definition~\ref{def:hsic:pop} in
Appendix~\ref{app:aux} and $\hat{\operatorname{HSIC}}$ is defined in
Definition~\ref{def:hsic:emp} in Appendix~\ref{app:aux}.

The proof proceeds in four steps. In the first two steps, we show uniform
convergence of \eqref{term:hsicn} to \eqref{term:hsicp} and \eqref{term:unifn}
to \eqref{term:unifp} and conclude that \eqref{term:lossn} converges uniformly
to \eqref{term:lossp}. In the third step, we show convergence in probability of
the estimated basis coefficients to the true $\parm^*$. In the final step, we
prove uniform convergence of the interventional CDFs.

\paragraph{Step 1.}
To show uniform convergence of \eqref{term:hsicn} to \eqref{term:hsicp}, we can
follow the argument in \citet[Appendix A.3.]{saengkyongam2022exploiting} to
conclude: 
\begin{itemize}
    \item Since $k_R$ and $k_Z$ are assumed to be non-negative and bounded, it
    follows by \citet[Corollary~15]{mooij2016hsic}, that for all $\parm \in \Theta$
    and $\epsilon > 0$,
    \[
    \lim_{n\to\infty} \Prob(\lvert H_n(\parm) - H(\parm)\rvert >
    \epsilon) = 0;
    \]
    \item By \citet[Lemma~16]{mooij2016hsic} it follows that for all $\parm, \parm'
    \in \Theta$ and $n \geq 2$, 
    \begin{align*}
    \onorm{H_n(\parm) - H_n(\parm')}
    &\leq \frac{32 L C}{\sqrt{n}} \norm{(a_{D_i}(Y_i)^\top(\parm -
    \parm'))_{i=1}^n}\\
    &\leq \frac{32 L C}{\sqrt{n}} \norm{(a_{D_i}(Y_i)^\top)_{i=1}^n} \cdot
    \norm{\parm - \parm'} \\
    &\leq \frac{32 L C}{\sqrt{n}} \sqrt{nA^2} \norm{\parm - \parm'} \\
    &= 32 L C A \norm{\parm - \parm'};
    \end{align*}
    where $L$ is the Lipschitz constant for $k_R(\Phi^{-1}(\bcd),
    \Phi^{-1}(\bcd))$, $C$ is an upper bound for $k_Z$ and $A$ satisfies
    $\sup_{d\in\calD}\sup_{y\in\calY} \lVert a_d(y) \rVert<A$, which exists by
    boundedness of $a_d$ for $d\in\calD$,
    \item $k_R$ and $\parm \mapsto F_{D}(Y; \parm)$ are (almost surely)
    continuously differentiable and $k_R$ and $k_Z$ have a bounded derivative by
    assumption. 
    Thus, by dominated convergence, the chain rule, the population
    representation of HSIC as in \citet[Proposition~2.5]{pfister2018hsic} and
    compactness of $\Theta$, it follows that $H(\parm)$ is Lipschitz on
    $\Theta$.
\end{itemize}
With the three points above, we can apply
\citet[Corollary~2.2]{newey1991uniform} to conclude that 
\begin{align}\label{eq:hsic:unif}
\lim_{n\to\infty} \Prob\left(\max_{\parm \in\Theta} \lvert H_n(\parm) - H(\parm)
\rvert > \epsilon\right) = 0,
\end{align}
where we additionally used that, by 
continuity of $\parm \mapsto |H_n(\parm)-H(\parm)|$ and compactness of $\Theta$
the function attains its maximum.

\paragraph{Step 2.}
It holds that:
\begin{itemize}
    \item For all $\parm\in\Theta$ and all $\epsilon>0$, $\lim_{n\to\infty}
    \Prob(\lvert U_n(\parm) - U(\parm) \rvert > \epsilon) = 0$
    (Lemma~\ref{lem:cip});
    \item $U_n$ and $U$ are Lipschitz (Lemma~\ref{lem:cvm:lip}).
\end{itemize}
We can conclude by \citet[Corollary 2.2]{newey1991uniform} 
that, for all $\epsilon > 0$,
\begin{align}\label{eq:unif:unif}
\lim_{n\to\infty} \Prob\left(\max_{\parm \in\Theta} \lvert U_n(\parm) - U(\parm)
\rvert > \epsilon\right) = 0.
\end{align}

Next, using that $\beta$ is Lipschitz continuous with constant $B$, we get
\begin{align}
&\lim_{n\to\infty} \Prob\left(\max_{\parm\in\Theta} \lvert
\ell_{\lambda,n}(\parm) - \ell_\lambda(\parm) \rvert > \epsilon\right)
\\
&\quad= 
\lim_{n\to\infty} \Prob\left(\max_{\parm\in\Theta} \lvert \beta(U_n(\parm),
\lambda H_n(\parm)) - \beta(U(\parm), \lambda H(\parm)) \rvert > \epsilon\right)
\\
&\quad\leq 
\lim_{n\to\infty} \Prob\left(\max_{\parm\in\Theta} \max\{ \onorm{U_n(\parm) -
U(\parm)}, \lambda \onorm{H_n(\parm) - H(\parm)}\} >
\epsilon/B \right)
\\
&\quad\leq 
\lim_{n\to\infty} \Prob\left(
\{
\max_{\parm\in\Theta} 
\onorm{U_n(\parm) -
U(\parm)} > \epsilon/B\} \cup
\{
\max_{\parm\in\Theta}
\lambda \onorm{H_n(\parm) - H(\parm)} > \epsilon/B \}\right)
\\
&\quad\leq
\lim_{n\to\infty} \Prob\left(
\max_{\parm\in\Theta} 
\onorm{U_n(\parm) -
U(\parm)} > \epsilon/B \right)
+ \Prob\left(
\max_{\parm\in\Theta}
\lambda \onorm{H_n(\parm) - H(\parm)} > \epsilon/B \right)\\
&\quad= 0,
\label{eq:loss:conclusion}
\end{align}
where for the last step we used \eqref{eq:hsic:unif} and \eqref{eq:unif:unif}.

\paragraph{Step 3.}
Now consider any sequence $\hat\parm^{n,\lambda} \in \argmin_{\parm \in \Theta}
\ell_{\lambda,n}(\parm)$. 
By Theorem~\ref{thm:id} (and since we assumed that
Assumption~\ref{asmp:relevance} and either Assumptions~\ref{asmp:pcu}
and~\ref{asmp:crs}~\ref{asmp:crs:crs} or Assumption~\ref{asmp:crs}
\ref{asmp:crs:rsim} are satisfied) it holds that $\parm^*$ is the unique
parameter for which $F_{D}(Y;\parm^*) \indep Z$ and $F_{D}(Y;\parm^*) \sim
\UD(0, 1)$. Moreover, by \citet{pfister2018hsic}, $H(\parm) = 0 \iff
F_{D}(Y;\parm) \indep Z$ and by Lemma~\ref{lem:cvm}, $U(\parm) = 0 \iff
F_{D}(Y;\parm) \sim \UD(0, 1)$. Hence, we get that $(U(\parm),
H(\parm))=(0,0)\iff \parm=\parm^*$, which together with $\beta(a, b) = 0 \iff
(a, b) = (0, 0)$ implies that $\parm^*$ is the unique minimizer of
$\ell_\lambda$ on $\Theta$.
Therefore, for all
$\epsilon > 0$, there exists $\zeta(\epsilon)$ such that for all $\parm \notin
B_\epsilon(\parm^*) \coloneqq \{\parm\in\Theta : \lVert\parm - \parm^*\rVert <
\epsilon\}$ it holds that $\ell_\lambda(\parm) - \ell_\lambda(\parm^*) >
\zeta(\epsilon)$.

Fix $\epsilon > 0$ and $\delta > 0$, then by \eqref{eq:loss:conclusion} there
exists $n' \in \mathbb{N}$, such that for all $n\in\mathbb{N}$ with $n \geq n'$,
\begin{align}\label{for:conclusion}
\Prob\left(\max_{\parm \in \Theta} \lvert\ell_{\lambda,n}(\parm) -
\ell_\lambda(\parm)\rvert > \frac{\zeta(\epsilon)}{2}\right) \leq
\frac{\delta}{2}.
\end{align}
Then, for all $n\in\mathbb{N}$ with $n \geq n'$,
\begin{align*}
\Prob(\lVert \hat\parm^{n,\lambda} - \parm^* \rVert > \epsilon) 
&\leq
\Prob\left(\min_{\parm\in\Theta\setminus B_\epsilon(\parm^*)}
\{\ell_{\lambda,n}(\parm) - \ell_{\lambda,n}(\parm^*)\} \leq 0\right) \\
&\leq
\Prob\left(\max_{\parm\in\Theta\setminus B_\epsilon(\parm^*)} \lvert
\ell_{\lambda,n}(\parm) - \ell_\lambda(\parm) \rvert + \lvert
\ell_{\lambda,n}(\parm^*) - \ell_\lambda(\parm^*) \rvert > \zeta(\epsilon)\right) \\
&\leq
\Prob\left(\left\{\max_{\parm\in\Theta\setminus B_\epsilon(\parm^*)} \lvert
\ell_{\lambda,n}(\parm) - \ell_\lambda(\parm) \rvert > \frac{\zeta(\epsilon)}{2}
\right\} \cup \left\{ \lvert
\ell_{\lambda,n}(\parm^*) - \ell_\lambda(\parm^*) \rvert >
\frac{\zeta(\epsilon)}{2} \right\}\right) \\
&\leq \delta,
\end{align*}
where
\begin{itemize}
    \item the first inequality follows because $\{\lVert\hat\parm^{n,\lambda} -
    \parm^* \rVert > \epsilon\}$ occurs only if there exists a $\parm \in \Theta
    \setminus B_\epsilon(\parm^*)$ such that $\{\ell_{\lambda,n}(\parm) -
    \ell_{\lambda,n}(\parm^*) \leq 0\}$ occurs;
    \item the second inequality follows because $\zeta(\epsilon)$ is such that
    for all $\parm \in \Theta \setminus B_\epsilon(\parm^*)$ we have
    $\ell_\lambda(\parm) - \ell_\lambda(\parm^*) > \zeta(\epsilon)$. Hence, if
    there exists $\parm \in \Theta$ such that $\{\ell_{\lambda,n}(\parm) -
    \ell_{\lambda,n}(\parm^*) \leq 0\}$ occurs, then
    there must exists $\parm \in\Theta$ such that
    $\{\lvert\ell_{\lambda,n}(\parm) - \ell_\lambda(\parm)\rvert + \lvert
    \ell_\lambda(\parm^*;
    \mathcal{D}_n) - \ell_\lambda(\parm^*) \rvert > \zeta(\epsilon)\}$ occurs;
    \item the third inequality follows because for
    $\{\max_{\parm\in\Theta\setminus B_\epsilon(\parm^*)} \lvert
    \ell_{\lambda,n}(\parm) -
    \ell_\lambda(\parm) \rvert + \lvert \ell_{\lambda,n}(\parm^*) -
    \ell_\lambda(\parm^*) \rvert > \zeta(\epsilon)\}$ to occur, either
    $\left\{\max_{\parm\in\Theta\setminus B_\epsilon(\parm^*)} \lvert
    \ell_{\lambda,n}(\parm) - \ell_\lambda(\parm) \rvert >
    \frac{\zeta(\epsilon)}{2} \right\}$ or $\left\{ \lvert
    \ell_{\lambda,n}(\parm^*) - \ell_\lambda(\parm^*) \rvert >
    \frac{\zeta(\epsilon)}{2} \right\}$ must occur;
    \item and the last inequality follows by the union bound together with
    \eqref{for:conclusion}.
\end{itemize}
Since $\delta>0$ was arbitrary, this implies
$\lim_{n\to\infty}\Prob(\lVert\hat\parm^{n,\lambda} - \parm^* \rVert >
\epsilon) = 0$.

\paragraph{Step 4.}
Finally, for all $d \in \calD$, by $\frac{1}{\sqrt{2\pi}}$-Lipschitz continuity
of $\Phi$ (which follows since $\Phi$ is differentiable everywhere and its
derivative is a density bounded above by $\tfrac{1}{\sqrt{2\pi}}$) and
boundedness of $a_d$, $d\in\calD$ by $A$, we obtain
\begin{align*}
    \sup_{y\in\calY} \onorm{\hat{F}_d^{n, \lambda}(y) - F_d^*(y)} &=
    \sup_{y\in\calY} \onorm{\Phi(a_d(y)^\top\hat\parm^{n,\lambda}) -
    \Phi(a_d(y)^\top\parm^*)} \\ &\leq \tfrac{1}{\sqrt{2\pi}}\sup_{y\in\calY}
    \onorm{a_d(y)^\top\hat\parm^{n,\lambda} - a_d(y)^\top\parm^*} \\ &\leq 
    \tfrac{1}{\sqrt{2\pi}}
    \sup_{y\in\calY} \norm{a_d(y)} \cdot
    \norm{\hat\parm^{n,\lambda} - \parm^*} \\ &\leq
    \tfrac{A}{\sqrt{2\pi}}
    \norm{\hat\parm^{n,\lambda} - \parm^*},
\end{align*}
which concludes the proof.
\end{proof}

\subsection{Proof of the statement in Example~\ref{ex:relevance}}
\label{proof:ex:relevance}

\begin{proof}
Fix arbitrary functions $m:\{0, 1\}\rightarrow\mathbb{R}$ and $k:\{0,
1\}\rightarrow\mathbb{R}$. For all $h, d\in \calD$ define $\gamma(h, d) = m(h) +
k(h) d$. Then, for $W=w(H, D)$ and $\mathcal{W}\coloneqq \{m(0), m(1), m(0) +
k(0), m(1) + k(1)\}$ it holds that
\begin{align}\label{eq:indepexample}
W \indep Z \iff \forall &w \in \mathcal{W}:\, \Prob(W = w \given Z = 0) = \Prob(W = w \given Z = 1).
\end{align}
Furthermore, explicitly computing the probabilities we get that for all $w \in
\mathcal{W}$ and all $z \in \calD$ it holds that
\begin{align*}
\Prob(W = w \given Z = z) &= 
0.5 \Prob(W = w \given Z = z, H = 0) + 
0.5 \Prob(W = w \given Z = z, H = 1) \\
&= 0.5 \big(p(0, z) \1(\gamma(0, 1) = w) + (1-p(0, z)) \1(\gamma(0, 0) = w)\big)
\\
&\quad\quad + 0.5 \big(p(1, z) \1(\gamma(1, 1) = w) + (1-p(1, z)) \1(\gamma(1,
0) = w)\big).
\end{align*}
Next, define for all $w\in\mathcal{W}$, $l_1(w) \coloneqq \1(\gamma(0,1) = w) -
\1(\gamma(0, 0) = w)$ and $l_2(w) \coloneqq \1(\gamma(1,1) = w) - \1(\gamma(1,
0) = w)$. Then, by \eqref{eq:indepexample}, we get that $W\indep Z$ is
equivalent to: for all $w\in\mathcal{W}$
\begin{align}\label{eq:final_condition}
    &p(0, 0) l_1(w) + \1(\gamma(0, 0) = w) +
p(1, 0) l_2(w) + \1(\gamma(1, 0) = w) \\
&\qquad =p(0, 1) l_1(w) + \1(\gamma(0, 0) = w) +
p(1, 1) l_2(w) + \1(\gamma(1, 0) = w).
\end{align}
Simplifying this expression we get that $W\indep Z$ is equivalent to for all
$w\in\mathcal{W}$
\begin{align*}
&(p(0, 0) - p(0, 1)) l_1(w) + 
(p(1, 0) - p(1, 1)) l_2(w)
= 0.
\end{align*}
Now, the functions $l_1$ and $l_2$ both take values in $\{-1,0,1\}$, hence since
by assumption $p(0,0)-p(0,1)\neq p(1,0)-p(1,1)$ and for all $z,h\in\{0,1\}$,
$p(z,h)\neq 1$ it holds that
\begin{equation*}
    \{(p(0, 0) - p(0, 1))l_1(w)\mid w\in\mathcal{W}\}\cap
     \{-(p(1, 0) - p(1, 1))l_2(w)\mid w\in\mathcal{W}\}\subseteq\{0\}.
\end{equation*}
Hence, the only way for \eqref{eq:final_condition} to be satisfied is if for all
$w\in\mathcal{W}$ it holds that $l_1(w)=l_2(w)=0$. This is however equivalent to
$\gamma(0,1)=\gamma(0,0)$ and $\gamma(1,0)=\gamma(1, 1)$, which is equivalent to
$k(0)=k(1)=0$. Therefore, we have shown that
\begin{equation*}
    Z\indep W \iff k\equiv 0,
\end{equation*}
which completes the proof.
\end{proof}

\section{Auxiliary definitions and results}\label{app:aux}

\begin{definition}[Population CvM criterion]\label{def:cvm:pop}
The population Cram\'er--von Mises criterion w.r.t.\ the standard uniform
distribution for a real-valued random variable $X$  
is defined as
$$
\operatorname{CvM}(X)\coloneqq \int_0^1 \left(\mathbb{P}(X\leq u) - u \right)^2
\dd u.
$$
\end{definition}

\begin{definition}[Empirical CvM criterion]\label{def:cvm:emp}
The empirical Cram\'er--von Mises criterion w.r.t.\ the standard uniform
distribution for $n$ i.i.d.\ real-valued random variables $(X_i)_{i=1}^n$ is
defined as
$$
\widehat{\operatorname{CvM}}((X_i)_{i=1}^n)\coloneqq \int_0^1 \left(\frac{1}{n}
\sum_{i=1}^n \1(X_i \leq u) - u \right)^2 \dd u =\frac{1}{12n^2} +
\frac{1}{n} \sum_{i=1}^n \left(\frac{i - 0.5}{n} - X_{(i)}\right)^2,
$$
where, for all $i\in\{1,\dots,n\}$, $X_{(i)}$ denotes the $i$th order statistic
and the equality is shown in Lemma~\ref{lem:halfrank}.
\end{definition}

\begin{definition}[Population HSIC]\label{def:hsic:pop}
Let $(Y, X)$ be random variables taking values in $\calY \times \calX \subseteq
\RR^{d_Y} \times \RR^{d_X}$ and $k : \calY \times \calY \to \RR_+$, $l : \calX
\times \calX \to \RR_+$ be two positive definite kernels. Let $(Y', X')$ and
$(Y'', X'')$ denote independent copies of $(Y, X)$. Then, the population HSIC is
defined as 
\begin{align}
\operatorname{HSIC}(Y, X; k, l)\coloneqq
\Ex[k(Y, Y')l(X, X')] + \Ex[k(Y, Y')]\Ex[l(X, X')] - 2 \Ex[k(Y, Y')l(X, X'')].
\end{align}
\end{definition}

\begin{definition}[Empirical HSIC]\label{def:hsic:emp}
Let $(Y_i, X_i)_{i=1}^n$ be independent and identically distributed random
variables taking values in $\calY \times \calX \subseteq
\RR^{d_Y}\times\RR^{d_X}$ and $k : \calY \times \calY \to \RR_+$, $l : \calX
\times \calX \to \RR_+$ be two positive definite kernels. Then, the empirical
HSIC is defined as
\begin{align}
&\hat{\operatorname{HSIC}}((Y_i)_{i=1}^n, (X_i)_{i=1}^n; k, l)
\\
&\quad
\coloneqq \ 
\frac{1}{n^2} \sum_{i,j}^n k(Y_i,Y_j)
l(X_i, X_j)
+
\frac{1}{n^4} \sum_{i,j,q,r}^n k(Y_i,Y_j)
l(X_q, X_r)
-
\frac{2}{n^3} \sum_{i,j,q}^n k(Y_i,Y_j)
l(X_i, X_q).
\end{align}
\end{definition}

\begin{lemma}\label{lem:cdfunique}
Let $Y$ be an absolutely continuous random variable and $F : \RR \rightarrow
[0,1]$ be a CDF. Then, it holds that $$F(Y) \sim \UD(0,1) \iff Y \sim F.$$
\end{lemma}

\begin{proof}
First, assume $Y\sim F$. For all $t\in(0,1)$ define the generalized inverse CDF
by $F^{-1}(t) \coloneqq \inf\{y \in \RR \mid F(y) \geq t\}$. Then, it holds for
all $t\in (0,1)$ that
\begin{equation}
\label{eq:cdf_parta}
    \Prob(F(Y) \leq t) = \Prob(Y \leq F^{-1}(t)) = F(F^{-1}(t)).
\end{equation}
Next, since $Y$ is absolutely continuous, we know that $F$ is continuous. For
all $t\in(0,1)$, by the intermediate value theorem there must exist
$y^*_t\in\mathbb{R}$ such that $F(y^*_t)=t$. Hence we get by the definition of
the generalized inverse that $F^{-1}(t)\leq y^*_t$, which further implies that
$F(F^{-1}(t))\leq F(y^*_t)=t$. Since it also holds that $F(F^{-1}(t))\geq t$, we
get that $F(F^{-1}(t))=t$. Together with \eqref{eq:cdf_parta} this implies that
$F(Y)\sim \UD(0,1)$.

Second, assume that $F(Y)\sim \UD(0,1)$. Then, it holds for all $y\in\mathbb{R}$
that
\begin{equation*}
    \Prob(Y\leq y) = \Prob(F(Y) \leq F(y)) = F(y)
\end{equation*}
and hence $Y\sim F$.
\end{proof}

\begin{definition}[Conditional independence, \citeauthor{dawid1979},
\citeyear{dawid1979}]\label{def:condind}
Let $X, Y, Z$ be random vectors. We say $X$ is conditionally independent of $Y$
given $Z$ and write $Y \indep X \given Z$ if and only if, for all bounded
measurable $\phi$,
\[
\Ex[\phi(Y) \given X, Z] = \Ex[\phi(Y) \given Z].
\]
\end{definition}

\begin{lemma}[Properties of conditional independence]\label{lem:ci}
Let $(A, B, C) \sim P$. Then, the following statements hold
\citep{dawid1980conditional}:
\begin{itemize}
    \item[] Decomposition: \[A \indep (B, C) \implies (A \indep B) \land (A
    \indep C).\]
    \item[] Weak union: \[A \indep (B, C) \implies (A \indep B \given C) \land
    (A \indep C \given B).\]
    \item[] Contraction: \[(A \indep B \given C) \land (A \indep C) \implies 
    A \indep (B, C).\]
\end{itemize}
\end{lemma}

\begin{lemma}\label{lem:cis1}
Under Setting~\ref{setting}, it holds that $(N_Y, H) \indep (D, Z) \given H$.
\end{lemma}
\begin{proof}
By Setting~\ref{setting}, we have
\[
(N_Z, N_D) \indep N_Y \given N_H.
\]
By \citet[Lemma~4.1]{dawid1979}, it holds
\[
(N_Z, N_D) \indep N_Y \given N_H \iff (N_Z, N_D, N_H) \indep (N_Y, N_H) \given N_H.
\]
Now, since $f$ from Setting~\ref{setting} is measurable, $H \coloneqq N_H$, $Z
\coloneqq N_Z$ and $D \coloneqq f(N_Z, H, N_D)$, we have
\begin{align}
(N_Z, N_D, N_H) \indep (N_Y, N_H) \given N_H 
&\implies
(N_Z, f(N_Z, N_D, N_H)) \indep (N_Y, N_H) \given N_H \\
&\implies
(Z, D) \indep (N_Y, H) \given H,
\end{align}
which concludes the proof.
\end{proof}

\begin{lemma}[Balance equation, \citeauthor{rosenbaum1983central},
\citeyear{rosenbaum1983central}]\label{lem:balancing}
In Setting~\ref{setting}, it holds that
\[
D \indep H \given \pi(H),
\]
where $\pi(H) \coloneqq \Ex[D \given H]$.
\end{lemma}
\begin{proof}
By
\[
\Prob(D = 1 \given H, \pi(H)) = \Prob(D = 1 \given H) = \pi(H),
\]
and
\[
\Prob(D = 1 \given \pi(H)) =
\Ex[\Ex[\1(D = 1) \given H, \pi(H)] \given \pi(H)]
= \Ex[\pi(H) \given \pi(H)] = \pi(H),
\]
we arrive at the desired conclusion.
\end{proof}

\begin{lemma}\label{lem:halfrank}
Let $X_1, \dots, X_n$ be independent and identically standard uniformly
distributed random variables and let $X_{(1)}, \dots, X_{(n)}$ denote their
order statistic. Then, it holds that
\[
\int_0^1 \left(\frac{1}{n} \sum_{i=1}^n \1(X_i \leq u) - u \right)^2 \dd u =
\frac{1}{12n^2} + \frac{1}{n} \sum_{i=1}^n \left(\frac{i - 0.5}{n} -
X_{(i)}\right)^2.
\]
\end{lemma}

\begin{proof}
    The result follows from the following computation
    \begin{align*}
       \int_0^1 \left(\frac{1}{n} \sum_{i=1}^n \1(X_i \leq u) - u \right)^2 \dd
       u
       &=\sum_{k=0}^n\int_{X_{(k)}}^{X_{(k+1)}}\left(\frac{1}{n} \sum_{i=1}^n
       \1(X_i \leq u) - u \right)^2 \dd u\\
       &=\sum_{k=0}^n\int_{X_{(k)}}^{X_{(k+1)}}\left(\tfrac{k}{n} - u \right)^2
       \dd u\\
       &=\sum_{k=0}^n\left[-\tfrac{1}{3}(\tfrac{k}{n}-X_{(k+1)})^3+\tfrac{1}{3}(\tfrac{k}{n}-X_{(k)})^3\right]\\
       &=\sum_{k=1}^n\tfrac{1}{3}\left[(\tfrac{k}{n}-X_{(k)})^3-(\tfrac{k-1}{n}-X_{(k)})^3\right]\\
       &=\sum_{k=1}^n\tfrac{1}{3}\left[(\tfrac{k}{n})^3-(\tfrac{k-1}{n})^3\right]
       + \left[\tfrac{1}{n}X_{(k)}^2-\tfrac{2k-1}{n}X_{(k)}\right]\\
       &=\sum_{k=1}^n\tfrac{1}{12n^3} +
       \tfrac{1}{n}\left(\tfrac{2k-1}{2n}-X_{(k)}\right)^2\\
       &=\tfrac{1}{12n^2} +
       \frac{1}{n}\sum_{k=1}^n\left(\tfrac{k-0.5}{n}-X_{(k)}\right)^2.
    \end{align*}
    This completes the proof of Lemma~\ref{lem:halfrank}.
\end{proof}

\begin{lemma}\label{lem:cip}
Let $X_1, \dots, X_n$ be independent and identically copies of $X$ taking values
in $\RR$. Then, it holds that
\[
\int_0^1 \left(\frac{1}{n} \sum_{i=1}^n \1(X_i \leq u) - u \right)^2 \dd u \;
\xrightarrow{\text{a.s.}} \;
\int_0^1 (\Prob(U \leq u) - u)^2 \dd u \quad\text{as $n\rightarrow\infty$}.
\]
\end{lemma}
\begin{proof}
For all $u \in (0, 1)$ and $n\in\mathbb{N}$, define $\hat F_n(u) \coloneqq
\frac{1}{n} \sum_{i=1}^n \1(X_i \leq u)$ and $F(u) \coloneqq \Prob(X \leq u)$.
We then have,
\begin{align}
&\quad\,\,
\int_0^1 \left(\hat F_n(u) - u \right)^2 \dd u
\\
&=
\int_0^1 \left(\hat F_n(u) - F(u) + F(u) - u \right)^2 \dd u
\\
&=
\int_0^1 (\hat F_n(u) - F(u))^2 + 
2 (\hat F_n(u) - F(u))(F(u) - u) +
(F(u) - u)^2
\dd u
\\
&=
\int_0^1 
(\hat F_n(u) - F(u))^2 \dd u + 
2 \int_0^1 
(\hat F_n(u) - F(u))(F(u) - u) \dd u
\\
&\quad\quad\,\, + \ 
\int_0^1 
(F(u) - u)^2 \dd u
\\
&\leq
\left(\sup_{u \in (0, 1)} \onorm{\hat{F}_n(u) - F(u)}\right)^2
+ 
\sup_{u \in (0, 1)} \onorm{\hat{F}_n(u) - F(u)}
\left(\int_0^1 (F(u) - u)^2 \dd u\right)^{1/2}
\\
&\quad\quad\,\,
+\ \int_0^1 (F(u) - u)^2 \dd u.
\\
&\leq
2 \sup_{u \in (0, 1)} \onorm{\hat{F}_n(u) - F(u)}
+ \int_0^1 (F(u) - u)^2 \dd u,
\label{eq:cvm:final}
\end{align}
where the last inequality uses that for all $u\in(0,1)$, $F(u)\in[0,1]$.
By the Glivenko--Cantelli theorem, it holds that, 
\(
\sup_{u\in(0,1)} \onorm{\hat F_n(u) - F(u)} \xrightarrow{\text{a.s.}} 0,
\)
which concludes the proof of Lemma~\ref{lem:cip}.
\end{proof}

\begin{lemma}\label{lem:cvm:lip}
Given the setting outlined in the proof of Proposition~\ref{prop:consistency},
there exists $K > 0$, such that, for all $\parm, \parm' \in\Theta$ and all
$n\in\mathbb{N}$,
\[
\onorm{U_n(\parm) - U_n(\parm')} \leq K \norm{\parm - \parm'}
\mbox{ and }
\onorm{U(\parm) - U(\parm')} \leq K \norm{\parm - \parm'}.
\]
\end{lemma}
\begin{proof}
For all $\parm\in\Theta$ and all $n\in\mathbb{N}$, it holds that
\begin{align}
\onorm{U_n(\parm) - U_n(\parm')} &=
\onorm{
\int_0^1 \left(\frac{1}{n} \sum_{i=1}^n \1(F_{D_i}(Y_i;\parm) \leq u) - u
\right)^2 
-
\left(\frac{1}{n} \sum_{i=1}^n \1(F_{D_i}(Y_i;\parm') \leq u) - u
\right)^2 \dd u
}
\\
&\leq
2 \int_0^1 
\onorm{
\frac{1}{n} \sum_{i=1}^n \1(F_{D_i}(Y_i;\parm) \leq u)
-
\1(F_{D_i}(Y_i;\parm') \leq u)
}
\dd u
\\
&\leq
2 \int_0^1 
\frac{1}{n} \sum_{i=1}^n
\onorm{
\1(F_{D_i}(Y_i;\parm) \leq u)
-
\1(F_{D_i}(Y_i;\parm') \leq u)
}
\dd u
\\
&\leq
\frac{2}{n} \sum_{i=1}^n
\int_0^1 
\onorm{
\1(F_{D_i}(Y_i;\parm) \leq u)
-
\1(F_{D_i}(Y_i;\parm') \leq u)
}
\dd u
\\ 
&\leq \label{eq:onward}
\frac{2}{n} \sum_{i=1}^n
\onorm{
F_{D_i}(Y_i;\parm)
-
F_{D_i}(Y_i;\parm')
}
\\
&=
\frac{2}{n} \sum_{i=1}^n
\onorm{
\Phi(a_{D_i}(Y_i)^\top\parm)
-
\Phi(a_{D_i}(Y_i)^\top\parm')
}
\\
&\leq
\frac{2}{\sqrt{2\pi}n} \sum_{i=1}^n
\onorm{
a_{D_i}(Y_i)^\top(\parm - \parm')
}
\\
&\leq
\frac{2}{\sqrt{2\pi} n} \sum_{i=1}^n
\norm{a_{D_i}(Y_i)^\top} \cdot \norm{(\parm - \parm')}
\\
&\leq
\frac{2 A}{\sqrt{2\pi}} \norm{(\parm - \parm')},
\end{align}
where the first inequality follows because the integrand is positive and for all
$x,y\in[0,1]$ it holds that $|x^2-y^2|<2|x-y|$; the second inequality follows by
the triangle inequality; the fourth inequality follows by evaluating the
integral; the fifth inequality follows from $\tfrac{1}{\sqrt{2\pi}}$-Lipschitz
continuity of $\Phi$; the sixth inequality follows from the Cauchy--Schwarz
inequality; and, the seventh inequality follows by $a_D(Y)$ being almost surely
upper bounded by $A$. 

Simiarly, for the popoulation $U$, it holds that
\begin{align}
\onorm{U(\parm)-U(\parm')} &=
\onorm{
\int_0^1 (\Prob(F_{D}(Y; \parm) \leq u) - u)^2 \dd u -
\int_0^1 (\Prob(F_{D}(Y; \parm') \leq u) - u)^2 \dd u
}
\\
&\leq
2\int_0^1 
\onorm{
\Prob(F_{D}(Y; \parm) \leq u) -
\Prob(F_{D}(Y; \parm') \leq u) 
}\dd u
\\
&=
2\int_0^1 
\onorm{
\Ex[\1(F_{D}(Y; \parm) \leq u) -
\1(F_{D}(Y; \parm') \leq u)]
}\dd u
\\
&\leq
2\Ex\left[
\int_0^1 
\onorm{
\1(F_{D}(Y; \parm) \leq u) -
\1(F_{D}(Y; \parm') \leq u)
}
\dd u
\right]
\\
&\leq
\frac{2 A}{\sqrt{2\pi}} \norm{(\parm - \parm')},
\end{align}
where the first inequality follows since for all $x,y\in[0,1]$ it holds that
$|x^2-y^2|<2|x-y|$; the second inequality follows by Jensen's inequality and the
last inequality follows analogously to \eqref{eq:onward} and onward.

This completes the proof of Lemma~\ref{lem:cvm:lip}.
\end{proof}

\begin{lemma}\label{lem:cvm}
Let $X$ be an absolutely continuous real-valued random variable. Then, it holds
that
\[
\int_0^1 (\Prob(X \leq u) - u)^2 \dd u = 0 \iff X \sim \UD(0, 1).
\]
\end{lemma}
\begin{proof}
The backward direction is immediate by $\Prob(X \leq u) = u$. For the
forward direction, 
\[
\int_0^1 (\Prob(X \leq u) - u)^2 \dd u = 0 \implies
\Prob(X \leq u) - u = 0 \mbox{ a.e.} \implies X \sim \UD(0, 1),
\]
where the first implication follows because the $L_2$ norm of a function is zero
if and only if the function zero almost everywhere and the second implication
follows by right-continuity of $\Prob(X \leq \bcd)$, absolute continuity of $X$
and the uniqueness of the CDF.
\end{proof}

\end{document}